\documentclass[line]{revtex4}
\usepackage{CJK}

\usepackage{epsfig}
\usepackage{amssymb}
\usepackage{amsmath}
\usepackage{latexsym}
\usepackage{makeidx}
\usepackage{graphicx}
\usepackage{verbatim}	
\usepackage{fancyhdr}
\usepackage{amsthm}
\usepackage{color} 

\newcommand{\vect}[1]{\mathbf #1}
\newcommand{\vecs}{\mathbf{s}}
\newcommand{\vecv}{\mathbf{v}}
\newcommand{\vecz}{\mathbf{z}}
\newcommand{\vecw}{\mathbf{w}}
\newcommand{\vecx}{\mathbf{x}}
\newcommand{\vecy}{\mathbf{y}}

\newcommand{\vecu}{\mathbf{u}}

\newcommand{\TAG}{\theta,\alpha,\gamma}
\newcommand{\TTAG}{\tilde\theta,\tilde\alpha,\tilde\gamma}

\newcommand{\balpha}{\boldsymbol\alpha}

\newcommand{\sbcII}{\mathcal B_{\mathbb C^2}}
\newcommand{\bs}[1]{\mathcal B_{\mathbb C^#1}}

\newcommand{\spn}[2]{*_{#2}}
\newcommand{\A}[2]{A_{(#1,#2)}^{(\hbar)}}	

\newcommand{\scsn}{\Phi_{\boldsymbol\alpha,\hbar}}
\newcommand{\csw}{\Phi_{\rho_{(n,m)}(\vecw)}^{(\hbar)}}
\newcommand{\csz}{\Phi_{\rho_{(n,m)}(\vecz)}^{(\hbar)}}
\newcommand{\cs}[1]{\Phi_{\rho_{(n,m)}(\vect {#1})}^{(\hbar)}}

\newcommand{\csalpha}{\Phi_{\boldsymbol\alpha}^{(\hbar)}}

\newcommand{\F}{\mathcal F}

\newcommand{\bts}[1]{\mathbf{B}_{S^#1}}

\newcommand{\e}{\mbox{\sl \Large{e}}\hspace{0.08cm}}

\newcommand{\Q}[1]{\mathbf{Q}_#1^{(\hbar)}}

\newcommand{\Qua}{\mathbb{Q}}

\newcommand{\B}[2]{{{\mathfrak B}_{(#1,#2)}^{(\hbar)}}}
\newcommand{\MB}[1]{{{B}_{#1}}}

\newcommand{\Be}{{{\mathfrak B}}}
\newcommand{\Sn}[1]{S^{#1}}
\newcommand{\T}{\mathrm T}
\newcommand{\TR}{\mathcal T}
\newcommand{\OL}{\mathrm L}
\newcommand{\du}[1]{\mathrm d \mu_{#1}^\hbar}
\newcommand{\p}{\mathfrak p}

\newtheorem{theorem}{Theorem}[section]
\newtheorem{corollary}[theorem]{Corollary}
\newtheorem{lemma}[theorem]{Lemma}
\newtheorem{proposition}[theorem]{Proposition}
\newtheorem{definition}[theorem]{Definition}
\newtheorem{remark}[theorem]{Remark}

\begin{document}

\begin {CJK*} {GB} { } % Use default fonts from CJK (see below)
\title {Star product induced from coherent states on $L^2(S^n)$ $n=2,3,5$}
\author {D\'iaz-Ort\'iz Erik Ignacio}
\affiliation {CONACYT-Universidad Pedag\'ogica Nacional-Unidad 201, Oaxaca, M\'exico}
\begin{abstract}
	We consider the bounded linear operators with domain in the Hilbert space $L^2(\Sn{n})$, $n=2,3,5$ and describe its symbolic calculus defined by the Berezin quantization.
	In particular, we derive an explicit formula for the composition of Berezin's symbols and thus  a noncommutative star product, which in turn is invariant under the action of the group $\mathrm{SU}(2)$, $\mathrm{SU}(2)\times \mathrm{SU}(2)$ and $\mathrm{SU}(4)$ on $\mathbb C^2$, $\mathbb C^4$ and $\mathbb C^8$ respectively.
\end{abstract}
\maketitle
\end{CJK*}

%\begin{keyword}
%Berezin quantization \sep Coherent states \sep star product \sep symbolic calculus \sep asymptotic expansion 
%\end{keyword}

\section{Introduction and summary}
The star product is an important tool for deformation quantization and noncommutative geometry. The most well studied star product is often referred to as Moyal star product. One way to obtain this star product is through the \textsl{Berezin quantization}. Let us briefly recall the definition of the latter (see Ref. \cite{B-74} for details)

Let $H$ be a Hilbert space endowed with the inner product  $(\cdot,\cdot)$ and $M$ be some set with the measure $\mathrm d\mu$. Let $\{\mathbf K(\cdot,\alpha) \in H\;|\; \alpha \in M\}$ be a family of functions in $H$
%, with the label corresponding to a point on $M$, 
such that it form an (overcomplete) basis for H.
\begin{comment}
satisfies the following properties:
\begin{enumerate}
\item[i) ] Let $\hat f(\alpha)= \big(f,\mathbf K(\cdot,\alpha)\big)$, the map $f \to \hat f$ is an embedding from $H$ into $L^2(M,\mathrm d\mu)$.
\item[ii) ] The family $\{\mathbf K(\cdot,\alpha)\}$ is complete: that is, for any $f,g \in H$,  Parseval's identity is valid
\begin{displaymath}
(f,g)= \int_M \big(f,\mathbf K(\cdot,\alpha)\big)\big (\mathbf K(\cdot,\alpha),g\big) \mathrm d \mu(\alpha)\;.
\end{displaymath}
\end{enumerate}
\end{comment}

The prototype of the space $H$ is a reproducing kernel Hilbert space of complex-valued holomorphic functions on a complex domain $M$, that is, for each $\alpha \in M$, there exists an
element $\mathbf K(\cdot,\alpha)$ of $H$, such that $(f,\mathbf K(\cdot,\alpha))=f(\alpha)$ for each $f\in H$.

In this case, one can directly define the Berezin symbol of a bounded linear operator $A$ on $H$ as the complex-valued function $\Be(A)$ on $M$ given by
\begin{equation*}\label{definition Berezin symbol}
\Be( A) (\alpha)= \frac{\big(A\mathbf K(\cdot,\alpha),\mathbf K(\cdot,\alpha)\big)}{\big(\mathbf K(\cdot, \alpha),\mathbf K(\cdot,\alpha)\big)}, \quad \alpha \in M.
\end{equation*}
and thus a symbolic calculus on $\mathbf B(H)$, the algebra of all bounded linear operators on $H$ \cite{B-74}.
%%%%%%%%%%%%%%%%%%%%%%%%%%%%%%%%%%%%%%%%%%%%%%%%%%%%%%%%%%%%%%%
%%%%%%%%%%%%%%%%%%%%%%%%%%%%%%%%%%%%%%%%%%%%%%%%%%%%%%%%%%%%%%%
%%%%%%%%%%%%%%%%%%%%%%%%%%%%%%%%%%%%%%%%%%%%%%%%%%%%%%%%%%%%%%%
\begin{comment}
Then we have a symbolic calculus on $\mathbf B(H)$, the algebra of all bounded linear operators on $H$. More specifically, we define the Berezin symbol of an operator $A \in \mathbf B(H)$ as the complex-valued function $\Be(A)$ on $M$ given by
\begin{displaymath}
\Be( A) (\alpha)= \frac{\big(A\mathbf K(\cdot,\alpha),\mathbf K(\cdot,\alpha)\big)}{\big(\mathbf K(\cdot, \alpha),\mathbf K(\cdot,\alpha)\big)}, \quad $\alpha \in M$,
\end{displaymath}
for each $\alpha \in M$ \cite{B-74}.
\end{comment}
%%%%%%%%%%%%%%%%%%%%%%%%%%%%%%%%%%%%%%%%%%%%%%%%%%%%%%%%%%%%%%%
%%%%%%%%%%%%%%%%%%%%%%%%%%%%%%%%%%%%%%%%%%%%%%%%%%%%%%%%%%%%%%%
%%%%%%%%%%%%%%%%%%%%%%%%%%%%%%%%%%%%%%%%%%%%%%%%%%%%%%%%%%%%%%%

The Berezin map $\Be : A \mapsto \Be(A)$ has some nice properties: is a linear operator, the unit operator corresponds to the unit constant, Hermitian conjugation of operators corresponds to complex conjugation of symbols. Moreover, if we assume that the Berezin symbol may be extended in a neighbourhood of the diagonal $M \times M$ to the function
\begin{displaymath}
\Be( A)(\alpha,\beta)= \frac{\big(A\mathbf K(\cdot,\alpha),\mathbf K(\cdot,\beta)\big)}{\big(\mathbf K(\cdot,\alpha),\mathbf K(\cdot,\beta)\big)},
\end{displaymath}
then an associative product for two Berezin symbols is defined by 
\begin{equation*}\label{Eq 2}
\Be(A) \# \Be(B)(\alpha) = 
\Be(AB)(\alpha)  = \int_M \Be(A)(\gamma,\alpha) \Be(B)(\alpha,\gamma) %\nonumber\\
%&\hspace{2cm}\times
\frac{\big|(\mathbf K(\cdot,\alpha),\mathbf K(\cdot,\gamma)\big)|^2}{\big(\mathbf K(\cdot,\alpha),\mathbf K(\cdot,\alpha)\big)} \mathrm d\mu(\gamma). 
\end{equation*}

The product $\#$ allows us to define a noncommutative star product on the algebra which consists of Berezin symbols for bounded linear operators with domain in $H$. 
%then we can define an associative product for two Berezin symbols (see Eq. ()) by
%we have the following formulas
\begin{comment}
\begin{align}
%\Be(A+cB) & = \Be( A) + c \Be( B)\;, \mbox{ for all constant } c \in \mathbb C,\nonumber\\
%\Be(\mathrm{Id}) & =1\;,\hspace{0.5cm} \mbox{where $\mathrm{Id}$ denotes the identity operator},\nonumber\\
%\Be( A^*)(\alpha,\beta) & = \overline{\Be(A)(\beta,\alpha)} \;, \hspace{0.5cm} \alpha,\beta\in M.\nonumber\\
\widehat{Af}(\alpha) & = \int_M \hat f(\beta) \Be(A)(\beta,\alpha) \big(\mathbf K(\cdot,\beta),\mathbf K(\cdot,\alpha)\big) \mathrm d \mu(\beta), \hspace{0.5cm} f \in H. \label{Eq 1}\\
\Be(AB)(\alpha,\beta) & = \int_M \Be(A)(\gamma,\beta) \Be(B)(\alpha,\gamma) %\nonumber\\
%&\hspace{2cm}\times
\frac{\big(\mathbf K(\cdot,\alpha),\mathbf K(\cdot,\gamma)\big)\big(\mathbf K(\cdot,\gamma),\mathbf K(\cdot,\beta)\big)}{\big(\mathbf K(\cdot,\alpha),\mathbf K(\cdot,\beta)\big)} \mathrm d\mu(\gamma). \label{Eq 2}
\end{align}
\end{comment}
%A useful application of this symbolic calculus is that it allows us to build a star product. 
In \cite{B-74} Berezin applied this method to K\"ahler manifold. In this case $H$ is the Hilbert space of analytic functions in $L^2(M,\mathrm d\mu)$ so that the embedding from $H$ into $L^2(M,\mathrm d\mu)$ is the inclusion, and the complete family $\{\mathbf K(\cdot,\alpha)\}$ is obtained by freezing one variable in the reproducing kernel.

The main goal of the present paper is to introduce a symbolic calculus for the Hilbert space $L^2(\Sn{n})$, $n=2,3,5$ of square integrable functions with respect to the normalized surface measure $\mathrm d \Omega$ on $\Sn{n}$ endowed with the inner product
\begin{equation*}
\langle \phi,\psi\rangle_{\Sn{n}}= \int_{\Sn{n}} \phi(\vect x) \overline{\psi(\vect x)} \mathrm d \Omega(\vect x),
\end{equation*} 
and $\Sn{n}=\{\vect x \in \mathbb R^{n+1}\;|\;|\vect x|=1\}$ is the unit sphere inmersed in $\mathbb R^{n+1}$.

In order to introduce this symbolic calculus, we first need  to find an overcomplete family of functions in $L^2(\Sn{n})$, $n=2,3,5$. Since it is not possible to define a reproducing kernel in $L^2(\Sn{n})$, we can not use the same idea applied by Berezin to K\"ahler manifold. % so we need to find another complete family in $L^2(\Sn{n})$. %  satisfying ii). 
However, in general, the coherent states are a specific complete set of vectors in a Hilbert space satisfying a certain resolution of the identity condition.%, i.e. the coherent states satisfy ii).

%Coherent states besides being mathematical tools  which provide a close connection between classical and quantum formalism, in general, they are a specific complete set of vectors in a Hilbert space satisfying a certain resolution of the identity condition, i.e. the coherent states satisfy ii). 
%Since it is not possible to define a reproducing kernel in the Hilbert space $L^2(\Sn{n})$, $n=2,3,5$ we need to found a family of functions in $L^2(\Sn{n})$  satisfying ii).

%\textcolor{blue}{
%In order to introduce this symbolic calculus,  we observe that the coherent states besides being mathematical tools  which provide a close connection between classical and quantum formalism, in general, they are a specific complete set of vectors in a Hilbert space satisfying a certain resolution of the identity condition, i.e. the coherent states satisfy ii). }

In \cite{V-02} Villegas-Blas introduced a system of coherent states for $L^2(\Sn{n})$, $n=2,3,5$. In addition, he defined an embedding $\bts{n}$ from $L^2(\Sn{n})$, $n=2,3,5$, to $\bs{m}$, $m=2,4,8$ respectively, where $\bs{m}$ denotes the Bargmann space of all entire functions in $L^2(\mathbb C^m,\du{m})$ for the Gaussian measure $\du{m}=(\pi \hbar)^{-m} \e^{-|\vect z|/\hbar} \mathrm d\vect z \mathrm d\overline{\vect z}$, with $\mathrm d\vect z \mathrm d \overline{\vect z}$ denoting the Lebesgue measure on $\mathbb C^m$ and $\hbar$ the Planck's constant %Moreover, the action of the Bargmann transform $\bts{n}$ on a function $\psi$ in $L^2(\Sn{n})$ is a function in $\F_m$ whose evaluation in $\vect z \in \mathbb C^m$ is equal to the $L^2(\Sn{n})$-inner product of $\psi$ with the coherent state labelled by $\vect z$, i.e. the Bargmann transform $\bts{n}$ satisfies i) 
(see section \ref{section bt + cs} for details on the notation and some general facts about the transformation $\bts{n}$ and coherent states for $L^2(\Sn{n})$, $n=2,3,5$).

Starting from this and adopting the approach along the lines of Berezin quantization, 
 %applying Berezin's theory, 
 in Sec. \ref{section Berezin symbolic calculus} we describe the rules for symbolic calculus on $\mathbf B(L^2(\Sn{n}))$, $n=2,3,5$. In particular, we derive an explicit formula for the composition of Berezin's symbols and thus a star product on the algebra which consists of Berezin symbols for bounded linear operators with domain in $L^2(\Sn{n})$. In addition, we show some properties of the extended Berezin symbol that we will use to obtain the asymptotic expansion of the star product.

In Sec. \ref{section star product} we will prove that this noncommutative star product satisfies the usual requirement on the semiclassical limit; this result can be obtained by using Laplace's method (see the Appendix \ref{Appendix A} for details of the way in which this method is used).%, in said appendix there are also mentioned other known results needed to prove the results of section \ref{section star product})

%see Appendix \ref{Appendix A} for details  Theorem \ref{theorem stationary phase method} of).

Finally, by the way in which Villegas-Blas introduced both the embedding $\bts{n}$ and the coherent states, we will prove in Sec. \ref{section invariance star product} the invariance of our star product under the action of the group $\mathrm{SU}(2)$, $\mathrm{SU}(2) \times \mathrm{SU}(2)$ and $\mathrm{SU}(4)$ on $\mathbb C^2$, $\mathbb C^4$ and $\mathbb C^8$ respectively.

It should be noted that, since $ H = L^2(\Sn{n})$ and $ L^2 (M, \mathrm d \mu) = L^2 (\mathbb C^n, \du{m}) $, in our construction there is no inclusion of $ H $ into $ L^2 (M, \mathrm d \mu) $. Furthermore, the functions of the complete family are not obtained by the reproducing kernel. This situation is thus slightly different Berezin's situation.

Throughout the paper, we will use the following  basic notation. For every $\vect z, \vect w \in \mathbb C^k$, $\vect z=(z_1,\ldots,z_k)$, $\vect w=(w_1,\ldots,w_k)$, and for every multi-index $\boldsymbol{ \ell}=(\ell_1,\ldots,\ell_k)\in \mathbb Z_+^k$ of length $k$, where $\mathbb Z_+$ is the set of nonnegative integers, let 
\begin{displaymath}
\vect z\cdot\vect w=\sum_{s=1}^k z_s\,\overline w_s, \hspace{1cm} |\vect z|=\sqrt{\vect z \cdot \vect z},\hspace{1cm}|\boldsymbol \ell|=\sum_{s=1}^k \ell_s, \hspace{1cm} \boldsymbol{\ell} !=\prod_{s=1}^k \ell_s! , \hspace{1cm} \vect z^{\boldsymbol \ell}=\prod_{s=1}^k z_s^{\ell_s}.
\end{displaymath}
Whenever convenient, we will abbreviate $\partial/\partial v_j, \partial/\partial \overline v_j$, etc., to $\partial_{v_j}, \partial_{\overline v_j}$, etc., respectively, and $\partial_{v_1} \partial_{v_2}\ldots \partial_{v_k}$ to $\partial_{v_1v_2\cdots v_k}$.

\section{Preliminaries}\label{section bt + cs}
In this section, we review some results on the transformation $\bts{n}$ and the system of coherent states for $L^2(\Sn{n})$, $n=2,3,5$ introduced by Villegas-Blas in \cite{V-02}. Here and in the sequel, the letters $n$ and $m$ will only denote integer numbers in the sets $\{2, 3, 5\}$ and $\{2, 4, 8\}$ respectively. Furthermore, whenever we write $(n,m)$ we mean the three possible cases $(n,m)= (2,2), (3,4), (5,8)$ unless a particular value of $(n, m)$ is specified.

In order to describe both the transformation $\bts{n}$ and the coherent states for $L^2(\Sn{n})$, $n=2,3,5$, let us consider, for $\vect z \in \mathbb C^m$, the map $\rho_{(n,m)}(\vect z)=(\rho_1(\vect z),\ldots, \rho_{n+1}(\vect z)) \in \mathbb C^m$ defined by

Case $(n,m)=(2,2)$ 
\begin{equation}\label{definition rho (2,2)}
\rho_1(\vect z)=\frac{1}{2}(z_2^2-z_1^2), \quad \rho_2(\vect z)=\frac{\imath}{2}(z_1^2+z_2^2), \quad \rho_3(\vect z)=z_1z_2.
\end{equation}

Case $(n,m)=(3,4)$
\begin{equation}\label{definition rho (3,4)}
\begin{aligned}
\rho_1(\vect z) & = z_1z_3+z_2z_4, & \rho_2(\vect z) & = \imath(z_1z_3-z_2z_4)\\
\rho_3(\vect z) & = \imath(z_1z_4+z_2z_3), & \rho_4(\vect z) & = z_1z_4-z_2z_3.
\end{aligned}
\end{equation}

Case $(n,m)=(5,8)$
\begin{equation}\label{definition rho (5,8)}
\begin{aligned}
\rho_1(\vect z) & = \imath(-z_1z_6+z_3z_8+z_2z_5-z_4z_7),\hspace{1cm}&
\rho_2(\vect z) & = z_1z_6+z_3z_8+z_2z_5+z_4z_7,\\
\rho_3(\vect z) & = z_2z_6+z_3z_7-z_1z_5-z_4z_8,&
\rho_4(\vect z) & = \imath(-z_1z_5+z_4z_8-z_2z_6+z_3z_7),\\
\rho_5(\vect z) & = \imath(-z_1z_8-z_2z_7-z_3z_6-z_4z_5),&
\rho_6(\vect z) & = z_1z_8+z_2z_7-z_3z_6-z_4z_5. 
\end{aligned}
\end{equation}

Notice that the range of the map $\rho_{(n,m)}$ is $\Qua^n:=\{\balpha\in \mathbb C^{n+1}\;|\; \alpha_1^2+\cdots+ \alpha_{n+1}^2=0\}$. Moreover, $\rho_{(n,m)}$ is not one to one because it is invariant under the action of the group $\mathcal G_m= \mathbb Z_2$, $\Sn{1}$ and $\mathrm{SU}(2)$ on $\mathbb C^2$, $\mathbb C^4$ and $\mathbb C^8$, respectively, %denoted by $\T(g)$ with $g$ in $\mathcal G_m$ as well and 
described by the following equation:
\begin{equation*}
\vect z'= \T(g)\vect z
\end{equation*}
with $\T(g)$ given by

Case $m=2$:
\begin{equation}\label{z'=T(g)z m=2}
\T(g)=\pm 1.
%\vecz'=\vecz\;\;\; \mbox{or}\;\;\; \vecz'=-\vecz. 
\end{equation}

Case $m=4$: For $\theta \in \mathbb R$ and therefore $g=\mathrm{exp}(\imath \theta) \in \Sn{1}$
\begin{equation}\label{z'=T(g)z m=4}
\T(g)=
%\vecz'=
\begin{pmatrix}
\exp(-\imath \theta) & 0 & 0 & 0\\
0 &\exp(-\imath \theta)  & 0 & 0\\
0 & 0 & \exp(\imath \theta)  & 0\\
0 & 0 & 0 & \exp(\imath \theta) \\
\end{pmatrix}.%\vecz.
\end{equation}

Case $m=8$: For $g \in \mathrm{SU}(2)$
\begin{equation}\label{z'=T(g)z m=8}
\T(g)=
%\vecz'= 
\mathbf L^\dagger \mathbf V(g) \mathbf L,% \vecz, 
\end{equation}
with $\mathbf L^\dagger$ denoting the adjoint of the following matrix
\begin{equation}\label{matrix L}
\mathbf L=
\begin{pmatrix}
1 & 0 & 0 & 0 & 0 & 0 & 0 & 0 \\
0 & 0 & 0 & 0 & 0 & 0 & 1 & 0 \\
0 & 0 & 1 & 0 & 0 & 0 & 0 & 0 \\
0 & 0 & 0 & 0 & 1 & 0 & 0 & 0 \\
0 & 0 & 0 & -1 & 0 & 0 & 0 & 0 \\
0 & 0 & 0 & 0 & 0 & 1 & 0 & 0 \\
0 & -1 & 0 & 0 & 0 & 0 & 0 & 0 \\
0 & 0 & 0 & 0 & 0 & 0 & 0 & 1 \\
\end{pmatrix}, \quad \mbox{and} \quad\mathbf V(g) =
\begin{pmatrix}
g & 0 & 0 & 0\\
0 & g & 0 & 0\\
0 & 0 & g & 0\\
0 & 0 & 0 & g
\end{pmatrix},
\end{equation}
where all the entries in $\mathbf V(g)$ are $2 \times 2$ matrices.

In fact, if we want the map $\rho_{(n,m)}$ to be a injection, its domain must be $\tilde{\mathbb C}^m/\mathcal G_m$, where
\begin{align}
\tilde{\mathbb C}^2 & = \mathbb C^2,\nonumber\\
\tilde{\mathbb C}^4 & = \Big\{\vect z \in \mathbb C^4\;:\; |z_1|^2+|z_2|^2=|z_3|^2+|z_4|^2\Big\}, \label{condition (3,4)}\\
\tilde{\mathbb C}^8 & = \Big\{\vect z \in \mathbb C^8\;:\; \sum_{j=1}^4|z_j|^2=\sum_{j=1}^4|z_{j+4}|^2, z_7 \overline z_1+z_5 \overline z_3=z_8 \overline z_2+z_6 \overline z_4\Big\}.\label{condition (5,8)}
%
%	\tilde{\mathbb C}^2 & = \Big\{\vect z \in \mathbb C^2\;:\;
%	\vect z = \T(g)(\vect z)\;\forall g \in \mathbb Z_2\Big\},\nonumber\\
%	\tilde{\mathbb C}^4 & = \Big\{\vect z \in \mathbb C^4\;:\; |z_1|^2+|z_2|^2=|z_3|^2+|z_4|^2, \vect z= \T(g)(\vect z)\;\forall g \in \Sn{1}\Big\}, \label{condition (3,4)}\\
%	\tilde{\mathbb C}^8 & = \Big\{\vect z \in \mathbb C^8\;:\; \sum_{j=1}^4|z_j|^2=\sum_{j=1}^4|z_{j+4}|^2, z_7 \overline z_1-z_8 \overline z_2+z_5 \overline z_3-z_6 \overline z_4=0,\nonumber\\
%	&\hspace{2.5cm} \vect z= \T(g)(\vect z)\;\forall g \in \mathrm{SU}(2)\Big\},\label{condition (5,8)}
\end{align}

%
%	Note that the canonical transformation $\mathcal C_{(n,m)}$ is not injective, since the map $\rho_{(n,m)}$ is invariant under the group actions of $\mathcal G_m= \mathbb Z_2, \Sn{1}, \mathrm{SU}(2)$ described bellow (see Eqs. (\ref{z'=T(g)z m=2}), (\ref{z'=T(g)z m=4}) and (\ref{z'=T(g)z m=8})). So, the domain 

\begin{comment}
Furthermore, Villegas-Blas \cite{V-02} observed that the function 
\begin{equation*}
\rho_{(n,m)}(\vect z) \cdot \vect x, \quad \vect z \in \mathbb C^m, \vect x \in \Sn{n},
\end{equation*}
is a generating function of the canonical transformation $\mathcal C_{(n,m)}$. From which, he defined a Bargmann type transform and coherent states for $L^2(\Sn{n})$, $n=2,3,5$. 
\end{comment}
%The rest of the section we will give details about the Bargmann transform and in next section we will give details about the coherent states.

\subsection{The transformation $\bts{n}$ with $n=2,3,5$}
Let $\bs{k}$ denote the Bargmann space of complex valued holomorphic functions on $\vect z \in \mathbb C^k$ which are square integrable with respect to the following measure on $\mathbb C^k$:
\begin{equation}\label{du}
\du{m}(\vect z)=\frac{1}{(\pi\hbar)^m} \e^{-|\vect z|^2/\hbar}  \mathrm d \vect z \mathrm d \overline {\vect z},
\end{equation}
with $|\vect z|^2=|z_1|^2+\cdots +|z_k|^2$ and $\mathrm d \vect z \mathrm d \overline {\vect z}$ the Lebesgue measure in $\mathbb C^k \simeq \mathbb R^{2k}$.

Let us consider the following transformation $\bts{n}$ from $L^2(\Sn{n})$, $n=2,3,5$ to $\bs{m}$, $m=2,4,8$ repectively 
\begin{equation}\label{Bargmann transform S^n}
\bts{n} \psi(\vect z)= \int_{\Sn{n}} \left(\sum_{\ell=0}^{\infty} \frac{\sqrt{2\ell+n-1}}{\ell !\sqrt{n-1}}       \left(\frac{\rho_{(n,m)}(\vect z)\cdot \vect x}{\hbar} \right)^\ell\right) \psi(\vect x) \mathrm d \Omega(\vect x),\quad \vect z \in \mathbb C^m,
\end{equation}
where $\hbar$ denotes the Planck's constant (regarded as a parameter).

Villegas-Blas proved that $\bts{n}$ is a unitary transformation onto its range $\F_m \subset \bs{m}$, where the space $\F_m$ is defined by

Case $m=2$: Let $\F_2$ be the closed subspace of $\bs{2}$ generated by the monomials with even degree.

Case $m=4$: Let $\F_4\subset \bs{4}$ be the kernel of the following operator:
\begin{equation}\label{Operator F4}
\mathcal L = v_1 \partial_{v_1}+ v_2\partial_{v_2} -v_3\partial_{v_3}- v_4\partial_{v_4}.
%	\mathcal L = v_1\frac{\partial}{\partial v_1}+ v_2\frac{\partial}{\partial v_2}-v_3\frac{\partial}{\partial v_3}- v_4\frac{\partial}{\partial v_4}.
\end{equation}

The domain of $\mathcal L$ is defined as $\{f \in \bs{4}\;|\; \mathcal L f \in \bs{4}\}$.

Case $m=8$: Let $\F_8\subset \bs{8}$ be the intersection of the kernel of the following three operators:
\begin{align}
\mathcal R_1 & = v_1\partial_{v_1} +v_2\partial_{ v_2}+ v_3\partial_{v_3}+ v_4\partial_{ v_4} -v_5\partial_{ v_5}-v_6\partial_{ v_6}-v_7\partial_{ v_7} -v_8\partial_{ v_8},\nonumber\\
\mathcal R_2 & = v_7\partial_{ v_1} -v_8\partial_{v_2} +v_5\partial_{ v_3}- v_6\partial_{v_4} -v_3\partial_{ v_5} +v_4\partial_{ v_6}- v_1\partial_{v_7} +v_2\partial_{ v_8},\label{Operators F8}\\
\mathcal R_3 & = z_7\partial_{v_1}-v_8\partial_{v_2} +v_5\partial_{ v_3}- v_6\partial_{v_4} +v_3\partial_{ v_5}- v_4\partial_{ v_6}+v_1 \partial_{v_7}- v_2\partial_{ v_8}. \nonumber
\end{align}

The domains of $\mathcal R_\ell$, $\ell=1,2,3$ are defined in a similar way as the domain of $\mathcal L$.

%with $\ker \mathcal A$ denoting the kernel of a given operator $\mathcal A$.
Notice that, for $m=2,4,8$, the elements of the Hilbert spaces $\F_m$ are given by the invariant functions in $\bs{m}$ under the action of the group $\mathcal G_m= \mathbb Z_2, S^1,\mathrm{SU}(2)$ on $\mathbb C^2,\mathbb C^4,\mathbb C^8$ respectively (see Eqs. (\ref{z'=T(g)z m=2}), (\ref{z'=T(g)z m=4}) and (\ref{z'=T(g)z m=8})).

\subsection{Coherent states for $L^2(\Sn{n})$, $n=2,3,5$}

In \cite{V-06} Villegas-Blas introduced the coherent states as the complex conjugate of the integral kernel defining the transformation $\bts{n}$ (see Eq. (\ref{Bargmann transform S^n})). For $\vect z \in \mathbb C^m-\{0\}$, let us define
\begin{equation}\label{coherent state}
\csz(\vect x)= \sum_{\ell=0}^{\infty} \frac{\sqrt{2\ell+n-1}}{\ell !\sqrt{n-1}} 
 \left(\frac{\vect x \cdot \rho_{(n,m)}(\vect z)}{\hbar}\right)^\ell, \quad \vect x \in \Sn{n}.
\end{equation}

Note that $\csz$ is in $L^2(\Sn{n})$ because it is a bounded function. Moreover, the action of the transformation $\bts{n}$ on a function $\Psi$ in $L^2(\Sn{n})$ is a function in $\bs{m}$ whose evaluation in $\vect z\in \mathbb C^m$ is equal to the $L^2(\Sn{n})$-inner product of $\Psi$ with the coherent state labeled by $\rho_{(n,m)}(\vect z)$ (see Eq. \ref{Bargmann transform S^n}).  This is
\begin{equation}\label{Bargmann = transform cs}
\bts{n} \Psi(\vect z)= \langle \Psi, \csz  \rangle_{\Sn{n}}.
\end{equation}

For $\vect w \in \mathbb C^m$, let us denote by $\Q{m}(\cdot,\vect w)$ to the action of the transformation  $\bts{n}$ of a coherent state $\csw$. The functions $\Q{m}(\cdot,\vect w)$  have the following expressions

\begin{align}
\Q{2}(\vect z,\vect w):=\bts{n} \Phi_{\rho_{(2,2)}(\vecw)}^{(\hbar)} (\vecz) & = \sum_{k=0}^\infty \frac{1}{(2k)!\hbar^{2k}}(z_1 \overline{w_1}+z_2\overline{w_2})^{2k} ,\nonumber\\
\Q{4}(\vect z,\vect w):=\bts{n} \Phi_{\rho_{(3,4)}(\vecw)}^{(\hbar)} (\vecz) & = \sum_{k=0}^\infty \frac{1}{(k!)^2\hbar^{2k}}(z_1 \overline{w_1}+z_2\overline{w_2})^k
(z_3 \overline{w_3}+z_4\overline{w_4})^k ,\label{Bargmann transform of a coherent state}\\
\Q{8}(\vect z,\vect w):=\bts{n} \Phi_{\rho_{(5,8)}(\vecw)}^{(\hbar)}(\vecz) & =  \sum_{k=0}^\infty  \frac{1}{k!(k+1)!\hbar^{2k}} \left(\varrho(\vecz,\vecw)\right)^k,\nonumber
\end{align}
with
\begin{equation}\label{def varrho(z,w)}
\begin{aligned}
\varrho(\mathbf u,\mathbf v) & =  [u_1 \overline{v_1}+u_2\overline{v_2} +
u_3 \overline{v_3}+u_4\overline{v_4}][u_5 \overline{v_5}+u_6\overline{v_6}+u_7\overline{v_7}+u_8\overline{v_8}]\\
&\; + [u_7 \overline{v_1}-u_8\overline{v_2}+ u_5 \overline{v_3}-u_6\overline{v_4}][
u_2\overline{v_8}-u_3 \overline{v_5}+
u_4\overline{v_6}-u_1\overline{v_7}].
\end{aligned}
\end{equation}

\begin{comment}
It can be shown that the Bargmann transform $\bts{n}$ of a coherent state is equal to the reproducing kernel of the space $\F_m$, i.e
\begin{equation}\label{btsn csw = rkf}
\bts{n} \csw (\vecz)=\Q{m}(\vecz,\vecw), \quad \vect z, \vect w \in \mathbb C^m.
\end{equation} 
%where $\Q{m}(\vecz,\vecw)$ is given in turn by the following expressions

%It can be shown that $\bts{n} \csw (\vecz)$ is equal to the reproducing kernel $\Q{m}(\vecz,\vecw)$ of the space $\F_m$ given in turn by the following expressions

%\textcolor{red}{
%\begin{remark}
%	Notice that for m = 2, 4, 8, the reproducing kernel $\Q{m}(\vecz, \vecw)$ is equal to the average of the integral kernel of $\mathcal B_m$ under the corresponding group $\mathbb Z_2$ , $S^1$ , $\mathrm{SU} (2)$ respectively.
% Moreover, the spaces $\F_m$ can be characterized as the elements of the Bargmann space $\mathcal B_m$ of $m$ complex variables which are invariant under the action of the group $Z_2$ , $S^1$ and $\mathrm{SU} (2)$ on $\mathbb C^2$, $\mathbb C^4$ or $\mathbb C^8$ respectively (see []).
%\end{remark}}

%\textcolor{red}{For $\vecw \in \mathbb C^m$ fixed, the reproducing kernel $\Q{m}(\vecz,\vecw)$ is a vector in the Hilbert space $\F_m$.} 
Using this equality, we can express the reproducing kernel $\Q{m}(\vecz,\vecw)$ in terms of the modified Bessel functions of the first kind of order $\nu$, $\mathbf I_\nu$ (see Secs. 8.4 and 8.5 of Ref. \cite{G-94} for definition and expressions for this special function) as the following proposition establishes it: 
\end{comment}

Moreover, we can express the functions $\Q{m}(\vecz,\vecw)$ in terms of the modified Bessel functions of the first kind of order $\nu$, $\mathbf I_\nu$ (see Secs. 8.4 and 8.5 of Ref. \cite{G-94} for definition and expressions for this special function) as the following proposition establishes it: 
\begin{proposition}\label{kernel F_m = Bessel}
	For $\vecz,\vecw \in \mathbb C^m-\{0\}:$
	\begin{equation}\label{btsn coherent states = bessel}
	\Q{m}(\vecz,\vecw)=\Gamma\left(\frac{n-1}{2}\right)\left(\frac{\balpha\cdot \boldsymbol \beta}{2\hbar^2}\right)^{\frac{3-n}{4}}\mathbf I_{\frac{n-3}{2}}\left(\frac{\sqrt{2\balpha\cdot \boldsymbol\beta}}{\hbar}\right),
	\end{equation}
	with $\balpha=\rho_{(n,m)}(\vecz)$ and $\boldsymbol{\beta}= \rho_{(n,m)}(\vecw)$. %\textcolor{red}{ and where $\mathbf I_\nu$ denotes the modified Bessel functions of the first kind of order $\nu$ (see Secs. 8.4 and 8.5 of Ref. \cite{G-94} for definition and expressions for this special function)}. 
	We are taking the branch of the square root function defined by $\sqrt z=|z|^{1/2} \mathrm{exp}(\imath \theta/2)$, where $\theta=\mathrm{Arg}(z)$ and $-\pi<\theta<\pi$.
\end{proposition}
\begin{proof}
	Using the explicit expression definition of $\rho_{(n,m)}$ (see Eqs. (\ref{definition rho (2,2)}), (\ref{definition rho (3,4)}) and (\ref{definition rho (5,8)})) we obtain the following relations 
	\begin{align}
	%\mathrm{Para }\;\; (n,m)=(2,2): & \;\;
	\rho_{(2,2)}(\vecz) \cdot \rho_{(2,2)}(\vecw) & = \frac{1}{2}(\vecz \cdot \vecw)^2,\nonumber\\
	%\mathrm{Para }\;\; (n,m)=(3,4): & \;\;
	\rho_{(3,4)}(\vecz) \cdot \rho_{(3,4)}(\vecw) & = 2(z_1\overline w_1+ z_2\overline w_2)(z_3\overline w_3+ z_4
	\overline w_4),\label{alpha beta = z w}\\[.15cm]
	%\mathrm{Para }\;\; (n,m)=(5,8): & \;\;
	\rho_{(5,8)}(\vecz) \cdot \rho_{(5,8)}(\vecw) & = 2\varrho(\vecz,\vecw),\nonumber
	\end{align}
	with $\varrho(\vecz,\vecw)$ defined in Eq. (\ref{def varrho(z,w)}). From Eqs. (\ref{Bargmann transform of a coherent state}), (\ref{alpha beta = z w}) and using the following expression by $\mathbf I_\nu$ (see formula 9.6.10 of Ref. \cite{AS-72})
	\begin{equation*}
	\mathbf I_\nu(\omega)=\left(\frac{\omega}{2}\right)^\nu \sum_{\ell=0}^{\infty}\frac{\left(\frac{1}{4}\omega^2\right)^\ell}{\ell! \Gamma(\nu+\ell+1)}\;
	\end{equation*} 
	we obtain Eq. (\ref{btsn coherent states = bessel}).
\begin{comment}
	\begin{equation*}
	\Q{m}(\vect z,\vect w)=\bts{n} \csw(\vecz)=\Gamma\left(\frac{n-1}{2}\right)\left(\frac{\balpha\cdot \boldsymbol \beta}{2\hbar^2}\right)^{\frac{3-n}{4}}\mathbf I_{\frac{n-3}{2}}\left(\frac{\sqrt{2\balpha\cdot \boldsymbol\beta}}{\hbar}\right).
	\end{equation*}
	From Eq. (\ref{btsn csw = rkf}) we conclude the proof of Proposition \ref{kernel F_m = Bessel}.
\end{comment}
\end{proof}

Proposition \ref{kernel F_m = Bessel} and Eqs.  (\ref{Bargmann = transform cs}), (\ref{Bargmann transform of a coherent state}) allow us to express the inner product of two coherent states in terms of the modified Bessel functions: for $\balpha=\rho_{(n,m)}(\vecz)$ and $\boldsymbol{\beta}= \rho_{(n,m)}(\vecw)$, with $\vect z,\vect w \in \mathbb C^m-\{0\}$

%Let $\balpha=\rho_{(n,m)}(\vecz)$ and $\boldsymbol{\beta}= \rho_{(n,m)}(\vecw)$, with $\vect z,\vect w \in \mathbb C^m-\{0\}$, from Eqs. (\ref{Bargmann = transform cs}), (\ref{Bargmann transform of a coherent state}) and (\ref{btsn coherent states = bessel})

%Let $\balpha=\rho_{(n,m)}(\vecz)$ and $\boldsymbol{\beta}= \rho_{(n,m)}(\vecw)$, with $\vect z,\vect w \in \mathbb C^m-\{0\}$, since the Bargmann transform $\bts{n}$ is a unitary transformation, we obtain from Eq. (\ref{btsn csw = rkf}), the reproducing property of $\Q{m}$ (see Eq. (\ref{reproducing property Qm})) and the expression for the reproducing kernel $\Q{m}$ given in Eq. (\ref{btsn coherent states = bessel})
%of a coherent state is equal to the reproducing kernel of $\F_m$ (see Eq. (\ref{btsn csw = rkf})), the unitarity of $\bts{n}$, the reproducing property of $\Q{m}$ and Eq. (\ref{btsn coherent states = bessel}) we have 
\begin{align}
\left\langle  \Phi_{\boldsymbol \beta}^{(\hbar)},  \Phi_{\balpha}^{(\hbar)}\right\rangle_{\Sn{n}} & = \Q{m}(\vect z, \vect w)\label{inner product cs = reproducing kernel}\\
& = \Gamma\left(\frac{n-1}{2}\right)\left(\frac{\balpha\cdot \boldsymbol \beta}{2\hbar^2}\right)^{\frac{3-n}{4}}\mathbf I_{\frac{n-3}{2}}\left(\frac{\sqrt{2\balpha\cdot \boldsymbol\beta}}{\hbar}\right).\label{inner product cs = Bessel}
\end{align} 

Moreover, from Eq. (\ref{inner product cs = Bessel}) and the fact that the modified Bessel function $\mathrm I_\vartheta$, $\vartheta\in\mathbb R$, has the following asymptotic expression when $|\omega| \to \infty$ (see formula 8.451-5 of Ref. \cite{G-94})
\begin{equation}\label{asymptotic expression modified Bessel function}
\mathrm I_\vartheta(\omega)=\frac{\mathrm e^\omega}{\sqrt{2\pi\omega}}\sum_{k=0}^\infty \frac{(-1)^k}{(2\omega)^k}\frac{\Gamma(\vartheta+k+\frac{1}{2})}{k!\Gamma(\vartheta-k+\frac{1}{2})}\;,\hspace{0.5cm} |\mathrm{Arg}(\omega)|<\frac{\pi}{2},
\end{equation}
we can obtain the asymptotic expansion for the inner product of two coherent states:%, as  the following proposition establishes it:
\begin{proposition}
	Let $\vect z,\vect w \in \mathbb C^m-\{0\}$ and $\balpha=\rho_{(n,m)}(\vect z), \boldsymbol \beta= \rho_{(n,m)}(\vect w)$. Assume $\balpha \cdot \boldsymbol \beta\ne 0$ and $|\mathrm{Arg}(\balpha \cdot \boldsymbol \beta )| < \pi$. Then
	%	$\rho_{(n,m)}(\vect z)\cdot \rho_{(n,m)}(\vect w) \ne 0$ and $|\mathrm{Arg}\big(\rho_{(n,m)}(\vect z)\cdot \rho_{(n,m)}(\vect w)\big)|<\pi$. Then
	\begin{equation*}
	\left\langle \Phi_{\boldsymbol\beta}^{(\hbar)},\Phi_{\boldsymbol\alpha}^{(\hbar)}\right\rangle_{\Sn{n}}=\Gamma\left(\frac{n-1}{2}\right) \left(\frac{\hbar^2}{2\boldsymbol\alpha \cdot \boldsymbol\beta}\right)^{\frac{n-2}{4}} \frac{2^{\frac{n-4}{2}}}
	{\sqrt \pi} \exp\left(\frac{\sqrt{2\boldsymbol\alpha\cdot\boldsymbol\beta}}{\hbar}\right)[1+\mathrm O(\hbar)]\;.
	\end{equation*}
\end{proposition}

\begin{comment}
From the last Proposition and the relation , we obtain a norm estimate for the coherent states:
\begin{align}
\|\scsn\|^2_{S^n} & = \Gamma\left(\frac{n-1}{2}\right) \left(\frac{\hbar^2}{2|\boldsymbol\alpha|^2}\right)^{\frac{n-2}{4}} \frac{2^{\frac{n-4}{2}}}
{\sqrt \pi} \exp\left(\frac{\sqrt{2}|\boldsymbol\alpha|}{\hbar}\right)[1+\mathrm O(\hbar)]\label{norma ec1}\\
& = \frac{2\pi^{n/2}}{\mathrm{Vol}(S^n)}\left[\frac{\hbar}{|\mathrm{Re}\boldsymbol\alpha|}\right]^{\frac{n}{2}-1}\frac{1}{n-1}\exp\left(\frac{ 2|\mathrm{Re}
\boldsymbol\alpha|}{\hbar}\right)[1+\mathrm O(\hbar)] \label{norm ecn}.
\end{align}
where we are use the equality $\sqrt 2 |\balpha|=|\vect z|^2$
\end{comment}

In addition to the expression indicated in Eq. (\ref{btsn coherent states = bessel}) for the functions $\Q{m}(\cdot, \vect w)$, with $\vect w \in \mathbb C^m$, we can also express these functions as the average of the map $\exp((\boldsymbol{\cdot})\cdot \vect w)$ under the corresponding group $\mathcal G_m= \mathbb Z_2, S^1,\mathrm{SU}(2)$ (see Eqs. (\ref{z'=T(g)z m=2}), (\ref{z'=T(g)z m=4}) and (\ref{z'=T(g)z m=8})), as the following proposition establishes it:
\begin{proposition}
	For $\vecz,\vecw \in \mathbb C^m-\{0\}:$
	\begin{align}
	\mathbf Q_2^{(\hbar)}(\vecz,\vecw) & = \frac{1}{2}\left(\exp\left(\frac{z_1 \overline{w_1}+z_2\overline{w_2}}{\hbar}\right)+
	\exp\left(\frac{-z_1\overline{w_1}-z_2\overline{w_2}}{\hbar}\right)\right) ,\nonumber\\
	\mathbf Q_4^{(\hbar)}(\vecz,\vecw) & = \frac{1}{2\pi} \int_{0}^{2\pi} \exp\left(\frac{1}{\hbar} \;\vecz \cdot
	\T(g(\psi))\vecw\right) \mathrm d\psi ,\label{integral expresion Q_m}\\
	\mathbf Q_8^{(\hbar)}(\vecz,\vecw) & = \int_{\theta=0}^{\pi/2}\int_{\alpha=0}^{2\pi}
	\int_{\gamma=0}^{2\pi}\exp\left(\frac{1}{\hbar}\;\vecz \cdot \T(g(\theta,\alpha,\gamma))\vecw\right) 
	\mathrm dm(\theta,\alpha,\gamma).\nonumber
	\end{align}
	with $\T(g)$ given by the action of $\mathcal G_m= \mathbb Z_2,\Sn{1}, \mathrm{SU}(2)$ on $\mathbb C^2,\mathbb C^4, \mathbb C^8$ indicated in Eqs. (\ref{z'=T(g)z m=2}), (\ref{z'=T(g)z m=4}) and (\ref{z'=T(g)z m=8}). We are using the notation $g=g(\psi)=\mathrm{exp}(\imath\psi)\in S^1$ for the $m=4$ case and $g=g(\theta,\alpha,\gamma)\in \mathrm{SU}(2)$ for the $m=8$ case where
	\begin{equation}\label{g parametrizada}
	g(\theta,\alpha,\gamma)=
	\begin{pmatrix}
	\cos(\theta)\exp(\imath \alpha) &\mathrm{sen}(\theta)\exp(\imath \gamma) \\[.1cm]
	-\mathrm{sen}(\theta) \exp(-\imath \gamma) & \cos(\theta) \exp(-\imath \alpha)
	\end{pmatrix}\;,\hspace{0.5cm}
	\end{equation}
	with  $\theta\in[0,\pi/2]$ and $\alpha,\gamma \in [-\pi,\pi]$. The measure appearing in the integral expression for $\mathbf Q_8(\vecz, \vecw)$ is the Haar measure of $\mathrm{SU} (2)$ under the parametrization of $\mathrm {SU} (2)$ indicated in Eq. (\ref{g parametrizada}): $\mathrm dm(\theta,\alpha,\gamma)=\frac{1}{2\pi^2}\;	\mathrm{sen}(\theta)\cos(\theta)
	\mathrm d\theta \mathrm d \alpha \mathrm d \gamma$.
\end{proposition}
\begin{proof}
	%The main idea is the same for the three cases $m=2,4,8$. We describe the most complicated case $m=8$. 
	Use the Taylor series expansion for the exponential function in the integral on the right-hand side of Eqs. (\ref{integral expresion Q_m}). The last case $(n,m)=(5,8)$ requires a more elaborated calculation. For this case, 	integrate first with respect to $\beta$ and then with respect to $\alpha$. To perform the integration with respect to $\theta$ use the identity $\int_0^{\pi/2} (\cos \theta)^{2k+1} (\sin \theta)^{2(\ell -k) +1} \mathrm d\theta = (\ell - k)! k! / 2(\ell + 1)!$, which in turn can be obtained using the formula 2.511-4 of Ref. \cite{G-94}.
\end{proof}

We end this section by showing that the family of coherent states is complete
\begin{proposition}\label{complete}
	The family of coherent states forms a complete system in $L^2(\Sn{n})$.
\end{proposition}
\begin{proof}
	Let $\phi,\psi \in L^2(\Sn{n})$, since the transformation $\bts{n}$ is unitary and Eq. (\ref{Bargmann = transform cs})
	\begin{equation*}
	\left \langle \phi,\psi\right\rangle_{\Sn{n}} = \left \langle \bts{n} \phi, \bts{n} \psi\right\rangle_{\F_m} = \int_{\mathbb C^m} \langle \phi, \csz\rangle_{\Sn{n}} \langle \csz, \psi \rangle_{\Sn{n}} \du{m}(\vect z).
	\end{equation*}
\end{proof}

\section{Berezin symbolic calculus}\label{section Berezin symbolic calculus}
In order to introduce our star product we consider the following:
%According to Berezin's theory (see Ref. \cite{B-74}), since the Bargmann trasform is a unitary operator, Eq. (\ref{Bargmann = transform cs})  and  Proposition \ref{complete}, we may consider the following
\begin{definition}
	The Berezin symbol of a bounded linear operator $A$ with domain in $L^2(\Sn{n})$ is defined, for every $\vecz \in \mathbb C^m$, by
	\begin{equation}\label{definition Berezin transform}
	\B{n}{m}(A)(\vecz)=\frac{\left\langle A\csz,\csz\right\rangle_{{\Sn{n}}}}{\left\langle \csz,\csz\right\rangle_{{\Sn{n}}}}\;.
	%\; \;\;\mbox{with } \;\balpha=\rho_{(n,m)}(\vecz).
	\end{equation}
	We  will denote by $\A{n}{m}$ the algebra of these Berezin symbols, that is
	\begin{displaymath}
	\A{n}{m}=\left\{\B{n}{m}(A)\;|\; A\in \mathbf B (L^2(\Sn{n}))\right\}.
	\end{displaymath}
\end{definition}

\noindent Example. Let's us consider the operators $\mathbf A_\ell$, $\ell=1,\ldots, n+1$, for which the coherent states $\csz$, $\vect z \in \mathbb C^m$, are their eigenfunctions with eigenvalue $(\overline{\rho_{(n,m)}(\vect z)})_\ell$ (see appendix \ref{appendix eigenfunctions cs} for details). From Eq. (\ref{definition Berezin transform}) it's not difficult to show
\begin{equation*}
\B{n}{m}(\mathbf A_\ell)(\vect z)= (\overline{\rho_{(n,m)}(\vect z)})_\ell\;, \quad \B{n}{m}(\mathbf A_\ell^*)(\vect z)= ({\rho_{(n,m)}(\vect z)})_\ell\;,
\end{equation*}
where $\mathbf A_\ell^*$ denotes the adjoint of the operator $\mathbf A_\ell$. 

Thus, by taking appropriate compositions of the operators $\mathbf A_\ell, \mathbf A_j^*$, $\ell,j=1,\ldots,n+1$, it can be shown that
\begin{align*}
\mathfrak E_{(n,m)}:=\left\{(\overline{\rho_{(n,m)}(\vect z)})^{\vect s} (\rho_{(n,m)}(\vect z))^{\vect k} \;|\; \vect k, \vect s \in \mathbb Z_+^{n+1}\right\} \subset \A{n}{m}.
\end{align*}
 
Moreover, since the map $\rho_{(n,m)}$ is invariant under the action of the group $\mathcal G_m = \mathbb Z_2$, $\Sn{1}$ and $\mathrm{SU}(2)$ on $\mathbb C^2$, $\mathbb C^4$ and $\mathbb C^8$, respectively (see Eqs. (\ref{z'=T(g)z m=2}), (\ref{z'=T(g)z m=4}) and (\ref{z'=T(g)z m=8})), we can prove that the space $\mathcal P_{inv}$ of all invariant polynomials, in the variables $\vect z$ and $\overline{ \vect z}$, under the action of the group $\mathcal G_m$ on $\mathbb C^m$ belong to $\A{n}{m}$. We will not go into great details of the proof and we will only sketch the structure of it. In order to show this assert, we first show that all polynomials in $\mathcal P_{inv}$ have even degree. 

For $\vect k,\vect s \in \mathbb Z_+^{m}$, let us define the function $P_{\vect k,\vect s}:\mathbb C^{m}\to \mathbb C$ by $P_{\vect k,\vect s}(\vect w):= \vect w^{\vect k} \overline{\vect w}^{\vect s}$.  If we assume that  $P_{\vect k,\vect s} \in \mathcal P_{inv}$, then $P_{\vect k,\vect s}=(P_{\vect k,\vect s})_{ave}$, with $(P_{\vect k,\vect s})_{ave}$ denoting the average of the polynomial $P_{\vect k,\vect s}$ under the group $\mathcal G_m$, i.e.
\begin{align}
P_{\vect k,\vect s}(\vect w) & =\frac{1}{2}\big(P_{\vect k,\vect s}(\vect w)+P_{\vect k,\vect s}(-\vect w)\big) \;,\quad for \; m=2,\nonumber\\
P_{\vect k,\vect s}(\vect w) & =\frac{1}{2\pi} \int_0^\pi P_{\vect k,\vect s}(\T(g(\psi))\vect w) \mathrm d\psi \;,\quad for \; m=4,\label{polynomial invariant case m=4}\\
P_{\vect k,\vect s}(\vect w) & =  \int_{\theta=0}^{\pi/2}\int_{\alpha=0}^{2\pi}
\int_{\gamma=0}^{2\pi} P_{\vect k,\vect s}(\T(g(\theta,\alpha,\gamma))\vect w) \mathrm dm(\theta,\alpha,\gamma) \;,\quad for \; m=8, \label{polynomial invariant case m=8}
\end{align}
with $\T(g)$ given by the action of $\mathcal G_m$ on $\mathbb C^m$ indicated in Eqs. (\ref{z'=T(g)z m=2}), (\ref{z'=T(g)z m=4}) and (\ref{z'=T(g)z m=8}). We assert that if $P_{\vect k, \vect s}$ is odd then it must be equal to the zero polynomial. 
%We only describe the most complicated case $m=8$. 
Using the explicit expression of $\T(g(\psi))$, $\T(g(\TAG))$ (see Eqs. (\ref{z'=T(g)z m=4}), (\ref{z'=T(g)z m=8})) on the right-hand side of Eqs. (\ref{polynomial invariant case m=4}), (\ref{polynomial invariant case m=8}) and integrating the result with respect to the variable $\psi$, $\theta$, respectively,  we obtain that  the right-hand side of Eqs. (\ref{polynomial invariant case m=4}), (\ref{polynomial invariant case m=8}) must be equal to zero, where for the case $m=8$ we have used the formulas 2.511-1, 2.511-3, 2.512-1 and 2.512-3 of Ref. \cite{G-94}.

Let $\ell \in \mathbb Z_+$, since the set of polynomials $\{P_{\vect k,\vect s}\;|\; \vect k,\vect s \in \mathbb Z_+^{m}\;,\; |\vect k|+|\vect s|=2 \ell\}$, is a basis of the space $\mathcal W_{2\ell}$ of homogeneous polynomials of degree $2 \ell$ in the complex variable $\vect w, \overline{\vect w} \in\mathbb C^{m}$, then the set of polynomials
$\left\{(P_{\vect k,\vect s})_{ave}\;|\; \vect k,\vect s \in \mathbb Z_+^{m}\;,\; |\vect k|+|\vect s|=2\ell\right\}$%, with $(P_{\vect k,\vect s})_{ave}$ denoting the average of the polynomial $P_{\vect k,\vect s}$ under the group $\mathcal G_m$, 
is a basis of the subspace of $\mathcal W_{2\ell}$ consisting of all invariant polynomials under the group $\mathcal G_m$. 
Moreover, by induction and algebraic manipulations it can be proved that $(P_{\vect k, \vect s})_{ave}$, $|\vect k|+|\vect s|=2\ell$, belongs to the span of the set $\mathcal S_{(n,m)}$ defined by
\begin{align*}
\mathcal S_{(2,2)} & = \left\{(z_1^2)^{a_1}(z_2^2)^{a_2} (z_1z_2)^{a_3}\;|\; a_j \in \mathbb Z_+,\; j=1,2,3,\; a_1+a_2+a_3=\ell\right\},\\
\mathcal S_{(3,4)} & =\big\{(z_1z_3)^{a_1}(z_2z_4)^{a_2} (z_1z_4)^{a_3}(z_2z_3)^{a_4}\;|\; a_j \in \mathbb Z_+, \;j=1,\ldots,4,\; a_1+\ldots+a_4=\ell\big\},\\
\mathcal S_{(5,8)} & = \big\{(z_3z_8+z_2z_5)^{a_1} (z_1z_6+z_4z_7)^{a_2} (z_2z_6+z_4z_8)^{a_3}(z_3z_7-z_1z_5)^{a_4} (z_1z_8+z_2z_7)^{a_5}(z_3z_6+z_4z_5)^{a_6} \;|\; a_j \in \mathbb Z_+,\\
& \hspace{0.5cm}\; j=1,\ldots,6,\; a_1+\ldots+a_6=\ell \big\}.
\end{align*}
Finally, since the functions in $\mathcal S_{(n,m)}$ belong to the span of the set $\mathfrak E_{(n,m)}$ we have $\mathcal P_{inv}$ belongs to $\A{n}{m}$.

From Eq. (\ref{inner product cs = Bessel}) we have 
\begin{equation*}
||\csz||_{\Sn{n}}^2= %||\Q{m}(\cdot,\rho_{(n,m)}(\vecz))||^2_{\F_m}=
\Gamma\left(\frac{n-1}{2}\right) \left(\frac{|\rho_{(n,m)}(\vect z)|^2} {2\hbar^2}\right)^{\frac{3-n}{4}} \mathrm I_{\frac{n-3}{2}}\left(\frac{\sqrt 2|\rho_{(n,m)}(\vecz)|}{\hbar}\right)>0,
\end{equation*}
hence the functions $\csz $ are continuous, i.e. the map $\vecz \mapsto |\csz|$ is continuous. Therefore, if $A:L^2(\Sn{n})\to L^2(\Sn{n})$ is a bounded linear operator, its Berezin symbol can be extended uniquely to a function defined on a neighbourhood of the diagonal in $\mathbb C^m \times \mathbb C^m$  in such a way that it is holomorphic in the first factor and anti-holomorphic in the second. In fact, such an extension is given explicitly by
\begin{equation}\label{extended covariant symbol}
\B{n}{m}(A)(\vecw,\vecz):=\frac{\langle A\csz,\csw\rangle_{{\Sn{n}}}}{\langle\csz,\csw\rangle_{{\Sn{n}}}}\;.
\end{equation}
\begin{remark}
	By Eq. (\ref{inner product cs = Bessel}), the extended Berezin symbol has singularities at the zeros of the modified Bessel function $\mathrm I_{\frac{n-3}{2}}(z)$, which are well known (see Ref. \cite{L-65} Sec. 5.13) and at $\rho_{(n,m)}(\vecw) \cdot\rho_{(n,m)}(\vecz)=0$. %Then the extended Berezin symbol is defined on $\mathbb C^n \times \mathbb C^n \setminus \mathcal S$, with
	%\begin{equation*}
	%\mathcal S=\left\{(\vecw,\vecz)\in \mathbb C^n \times \mathbb C^n\;|\; \vecw \cdot\vecz=0 \mbox{ or } \frac{2\sqrt{\vecw\cdot\vecz}}{\hbar}=\imath \lambda \mbox{ with } \lambda\in \mathbb R,\; \mathrm I_{n+p-1}(\imath\lambda)=0 \right\}.
	%\end{equation*}
	%where $\lambda$ is a negative real  number that satisfies  $\mathrm I_{n+p-1}(\imath\lambda)=0$.
\end{remark}

%From Eq. (\ref{extended covariant symbol})  some nice properties of the extended Berezin symbol are followed directly: it is linear, the identity operator corresponds to the unit constant and Hermitian conjugation of operators corresponds to complex conjugation of symbols. Moreover, 

We now give the rules for symbolic calculus

\begin{proposition}\label{proposition rules for Berezin}
%	\begin{enumerate}
%		\item [a)] The Berezin symbol 
%		\item [c)] 
Let $A, B$ be bounded linear operators with domain in $L^2(\Sn{n})$. Then for $\vecz,\vecw \in \mathbb C^m$ and $\phi\in L^2(\Sn{n})$ we have
		\begin{align}
		\B{n}{m}(\mathrm{Id})& = 1\;, \mbox{with $\mathrm{Id}$ denoting the identity operator,}\nonumber\\
		\B{n}{m}(A^*) (\vect z,\vect w) & = \overline{\B{n}{m}(A)(\vect w,\vect z)}\;,\nonumber\\
		\bts{n}(A\phi)(\vect z) & = %\frac{1}{(\pi\hbar)^m}
		\int_{\mathbb C^m} \bts{n} \phi(\vect u) \B{n}{m}(A)(\vect z,\vect u) \Q{m}(\vect z,\vect u)\du{m}(\vect u).\label{Eq 4}\\
		\B{n}{m}(AB)(\vect z,\vect w) & = \int_{\mathbb C^m}  \B{n}{m}(B)(\vect u,\vect w) \B{n}{m}(A)(\vect z,\vect u) %\nonumber\\
		%& \hspace{2cm} \times 
		\frac{ \Q{m}(\vect z,\vect u) \Q{m}(\vect u,\vect w)}{\Q{m}(\vect z,\vect w)}\; \du{m}(\vect u)%\mathrm d\vecu \mathrm d\overline\vecu\;
		,\label{Eq 3}
		\end{align}
%	\end{enumerate}	
\end{proposition}
\begin{proof}
	The first two equalities follow from the definition of Berezin symbol (see Eq. (\ref{definition Berezin transform})). Eq. (\ref{Eq 4}) follows from 
	Eq. (\ref{Bargmann = transform cs}), Proposition \ref{complete} and the definition of extended Berezin symbol (see Eq. (\ref{extended covariant symbol}))
	\begin{equation*}
	\bts{n}(A \phi)(\vect z)= \left\langle \phi, A^*\cs{z}\right\rangle_{\Sn{n}}= \int_{\mathbb C^m} \bts{n} \phi(\vect u) \B{n}{m}(A)(\vect z,\vect u) \Q{m}(\vect z,\vect u)\du{m}(\vect u).
	\end{equation*}
	The formula in Eq. (\ref{Eq 3}) follows from Proposition \ref{complete} and Eq. (\ref{extended covariant symbol}).
%	This is a direct consequence of the formulas in Ref. \cite{B-74}, the unitarity of the Bargmann transform $\bts{n}$, Eq. (\ref{btsn csw = rkf}) and the reproducing property of $\Q{m}$ (see Eq. (\ref{reproducing property Qm})). %the equality $\bts{n}\Phi_{\rho_{(n,m)}(\mathbf v)}^{(\hbar)}(\mathbf s) =\Q{m}(\mathbf s,\mathbf v)$ for all $\mathbf v, \mathbf s \in \mathbb C^m$.
	%\textcolor{red}{The proof  duplicates the standard proof of the existence of a Bergman kernel function}
\end{proof}

\begin{corollary}\label{corollary}
	Let $A:L^2(\Sn{n})\to L^2(\Sn{n})$ be a bounded linear operator, and define the function on $\mathbb C^m \times \mathbb C^m$,
	\begin{displaymath}
	\mathfrak K_A(\vecz,\vecw):=\big\langle A\csw,\csz\big\rangle_{\Sn{n}}.
	\end{displaymath}
%	Let $A^{\sharp}$ be the operator on $L^2(\mathbb C^n,\du{m})$ with Schwartz kernel $\mathfrak K_A$.
%	Then $\bts{n} A(\bts{n})^{-1}$ is a integral operator on $\F_m$ with Schwartz kernel $\mathfrak K_A$. Namely, for all $f\in \F_m\;,\;\vecz \in \mathbb C^m$ we have
%	\begin{equation} \label{BAB-1=integral operator}
	%\bts{n} A(\bts{n})^{-1} f(\vecz)=\int_{\mathbb C^m} f(\vecw) \mathfrak K_A(\vecz,\vecw) \du{m}(\vecw).
	%\end{equation}
	Let $A^\sharp$ be the operator on $\F_m$ with Schwarz kernel $\mathfrak K_A$. Then
		\begin{equation}\label{BAB-1=integral operator}
		A^\sharp f(\vecz)=\bts{n} A(\bts{n})^{-1} f(\vecz) \;\hspace{0.5cm}\forall f\in \F_m\;,\;\vecz \in \mathbb C^m.
		\end{equation}
\end{corollary}
\begin{proof} 	
	Let $f \in \F_m$, and $\phi=(\bts{n})^{-1} f$. We obtain Eq. (\ref{BAB-1=integral operator}) from Eqs. (\ref{Eq 4}),  (\ref{extended covariant symbol}) and (\ref{inner product cs = reproducing kernel}).
	%Then from Eqs. (\ref{Eq 4}) and (\ref{extended covariant symbol})
	%\begin{equation*}
	%\bts{n} A(\bts{n})^{-1} f(\vecz) = \int_{\mathbb C^m} f(\vect w) \mathfrak K_A(\vect z,\vect w) \du{m}.
	%\end{equation*}
\begin{comment}
	\begin{align*}
	A^\sharp f(\vect z) & = \int_{\mathbb C^m} \left\langle \phi, \csw \right\rangle_{\Sn{n}}\left\langle A\csw,\csz\right\rangle_{\Sn{n}} \du{m}(\vect w)\nonumber\\
	& = \left\langle\int_{\mathbb C^m} \left\langle \phi, \csw \right\rangle_{\Sn{n}}A\csw \du{m}(\vect w), \csz\right\rangle_{\Sn{n}}.
	\end{align*}
	Applying $A$ to both sides of (), we see that the first entry in the last inner product is $A(\psi)$, and therefore
	
	We obtain Eq. (\ref{BAB-1=integral operator}) from Eqs. (\ref{Eq 4}),  (\ref{extended covariant symbol}) and (\ref{inner product cs = reproducing kernel}).%  (\ref{btsn csw = rkf}) and (\ref{Bargmann = transform cs}).
	%	\begin{align*}
	%	\bts{n} A (\bts{n})^{-1} f (\vecz)= \int_{\mathbb C^m} f(\vecw) \mathfrak K_A(\vecz,\vecw) \du{m}(\vecw).
	%	\end{align*}
\end{comment}	
\end{proof}

Now we show some properties of the extended Berezin symbol that we will use in the next section to obtain the asymptotic expansion of the star-product.

\begin{proposition}\label{proposition extended berezin symbol invariant}
	Let $A$ be a bounded linear operator on $L^2(\Sn{n})$, $n=2,3,5$ and $\vecz,\vecw \in \mathbb C^m$, $m=2,4,8$, respectively. Then
	\begin{equation}\label{extended berezin symbol invariant}
	\B{n}{m}(A)\big(\T(g)\vecw,\T(\tilde g)\vecz\big)=\B{n}{m}(A)(\vecw,\vecz)\;, \hspace{1cm}\forall g,\tilde g \in \mathcal G_m
	\end{equation}
	where  $\T(g), \T(\tilde g)$ are given by the action of $\mathcal G_m$ %=\mathbb Z_2, S^1, \mathrm{SU}(2)$
	on $\mathbb C^m$%, \mathbb C^4,\mathbb C^8$ respectively 
	(see Eqs. (\ref{z'=T(g)z m=2}), (\ref{z'=T(g)z m=4}) and (\ref{z'=T(g)z m=8})). 
	\begin{comment}
	In particular for $\vect z$ fixed, the extended Berezin symbol $\B{n}{m}(A)(\vecw,\vecz)$ is invariant under the action of the group $\mathcal G_m$, i.e
	\begin{equation}\label{ebs invariant one variable}
	\B{n}{m}(A)(\T(g)\vecw,\vecz)=\B{n}{m}(A)(\vecw,\vecz)\;, \hspace{0.5cm}\forall \vect w \in \mathbb C^m\;,\; \forall g, \in \mathcal G_m\;.
	\end{equation}
	A similar result is obtained by freezing the first variable in the extended Berezin symbol.
	\end{comment}
\end{proposition}
%\begin{theorem}
%	Let $A$ be a bounded linear operator on $L^2(\Sn{n})$. Then for $\vecw\in \mathbb C^m$ fixed,  the extended Berezin symbol $\B{n}{m}(A)(\vecz,\vecw)$ is invariant under the action of the group $Z_2$, $S^1$ and $\mathrm{SU}(2)$ on $\mathbb C^m$ for $m=2,4,8$ respectively, i.e. $\B{n}{m}(\T(g)\vecz,\vecw)=\B{n}{m}(\vecz,\vecw)$, where $\T(g)$ is given by the action of $\mathcal G_m$ on $\mathbb C^m$ indicated in Eqs. (), () and ().
%\end{theorem}
\begin{proof}
	%Let $g$ in $Z_2, S^1,\mathrm{SU}(2)$ and $\T(g)$ the action on $\mathbb C^m$ for $m=2,4,8$ indicated in Eqs (), () and () respectively. 
	Since the elements of the Hilbert spaces $\F_m$ are  invariant functions in $\bs{m}$ under the action of the corresponding group $\mathcal G_m$, and $\Q{m}(\boldsymbol \cdot,\vecv)=\bts{n} \cs{\vect v} \in \F_m$, $\vect v \in \mathbb C^m$ (see Eqs. (\ref{Bargmann = transform cs}) and \ref{inner product cs = reproducing kernel})), we have
	%Since, for $\vect v \in \mathbb C^m$ fixed, the function $\Q{m}(\vecs,\vecv)=\bts{n} \cs{\vect v}(\vect s)$ is a vector in the space $\F_m$ and %in turn satisfies the equality  $\Q{m}(\vecs,\vecv)= \overline{\Q{m}(\vecv,\vecs)}$, then:
	\begin{equation}\label{csf invariant second variable}
	\Q{m}(\vecs,\T(g)\vecv)=\Q{m}(\vecs,\vecv), \hspace{0.5cm} \forall \vecs,\vecv \in \mathbb C^m,\;\forall g \in \mathcal G_m\;,
	\end{equation}
	where we have used $\Q{m}(\vecs,\vecv)= \overline{\Q{m}(\vecv,\vecs)}$. Thus, from  Corollary \ref{corollary} and Eq. (\ref{csf invariant second variable})
	\begin{equation}\label{tbsn A tbsn inverse invariant}
	\bts{n} A (\bts{n})^{-1} \Q{m}(\cdot,\T(g)\vecs)=\bts{n} A (\bts{n})^{-1} \Q{m}(\cdot,\vecs),\quad \forall \vect s \in \mathbb C^m.
	\end{equation}
	Since the transformation $\bts{n}$ is unitary, we obtain from Eq. (\ref{extended covariant symbol})% and (\ref{btsn csw = rkf})
	\begin{equation}\label{extended berezin symbol + csf}
	\B{n}{m}(A)(\vecw,\vecz)= \frac{\big \langle \bts{n} A (\bts{n})^{-1} \Q{m}(\cdot,\vecz), \Q{m}(\cdot,\vecw)\big\rangle_{\F_m} }{\big\langle\Q{m}(\cdot,\vecw),\Q{m}(\cdot,\vecz)\big\rangle_{\F_m}}\;.
	\end{equation}
	
	Therefore, from Eqs. (\ref{extended berezin symbol + csf}), (\ref{tbsn A tbsn inverse invariant}) and (\ref{csf invariant second variable}) we obtain Eq. (\ref{extended berezin symbol invariant}).
	
%	On the other hand, let $g_0$ be the identity element in $\mathcal G_m$. From the explicit expression for $\T(g)$ (see Eqs. (\ref{z'=T(g)z m=2}), (\ref{z'=T(g)z m=4}) and (\ref{z'=T(g)z m=8})), we obtain $\T(g_0)\vecv=\vecv$ for all $\vecv \in \mathbb C^m$. Thus, taking $\tilde g=g_0$ in Eq (\ref{extended berezin symbol invariant}), we obtain Eq. (\ref{ebs invariant one variable}).
\end{proof}

\begin{corollary}\label{extended berezin symbol invariant one variable}
	Let $A$ be a bounded linear operator on $L^2(\Sn{n})$, $n=2,3,5$, then for $\vecz\in \mathbb C^m$ fixed, $m=2,4,8$ respectively,  the extended Berezin symbol $\B{n}{m}(A)(\vecw,\vecz)$ is invariant under the action of the group $\mathcal G_m$ %\mathbb Z_2$, $S^1$ and $\mathrm{SU}(2)$
	on $\mathbb C^m$ (see Eqs. (\ref{z'=T(g)z m=2}), (\ref{z'=T(g)z m=4}) and (\ref{z'=T(g)z m=8})), %$\mathbb C^m$ for $m=2,4,8$ respectively,
	i.e. 
	\begin{equation}\label{ebs invariant one variable}
	\B{n}{m}(A)(\T(g)\vecw,\vecz)=\B{n}{m}(A)(\vecw,\vecz)\;,\hspace{0.5cm}\forall \vecw \in \mathbb C^m\;,\; \forall g \in \mathcal G_m\;.
	\end{equation}
	%for $g \in \mathcal G_m$: $\B{n}{m}(A)(\T(g)\vecw,\vecz)=\B{n}{m}(A)(\vecw,\vecz)$. %, where $\T(g)$ is given by the action of $\mathcal G_m$ on $\mathbb C^m$ indicated in Eqs. (), () and (). 
	A similar result is obtained by freezing the first variable in the extended Berezin symbol.
\end{corollary}
\begin{proof}
	Let $g_0$ be the identity element in $\mathcal G_m$. From the explicit expression for $\T(g)$ (see Eqs. (\ref{z'=T(g)z m=2}), (\ref{z'=T(g)z m=4}) and (\ref{z'=T(g)z m=8})), we obtain $\T(g_0)\vecv=\vecv$ for all $\vecv \in \mathbb C^m$. Thus, taking $\tilde g=g_0$ in Eq (\ref{extended berezin symbol invariant}), we obtain Eq. (\ref{ebs invariant one variable}). 
	
	Similarly, we obtain the second part of this Corollary.
\end{proof}

In addition to satisfying  the property indicated in Proposition \ref{proposition extended berezin symbol invariant}, the extended Berezin symbol belongs to the kernel of the operator $\mathcal L$ (for the case $(n,m)=(3,4)$) and the kernel of the operators $\mathcal R_1$, $\mathcal R_2$, $\mathcal R_3$ (for the case $(n,m)=(5,8)$), as the following proposition establishes it: 
\begin{proposition}\label{extended Berezin symbol in kernel}
	Let $A$ be a bounded linear operator on $L^2(\Sn{n})$, $n=3,5$, and $\vect{w},\vect{v} \in \mathbb C^m$, $m=4,8$ respectively. Assume $\rho_{(n,m)}(\vect{v}) \cdot\rho_{(n,m)} (\vect{w})\ne 0$ and $\mathrm{Arg}(\rho_{(n,m)}\big(\vect{v})\cdot\rho_{(n,m)}(\vect{w})\big)<\pi$. Then
	\begin{align*}
	\mbox{For $(n,m)=(3,4):$} & \hspace{1.5cm}  \mathcal L\;\B{3}{4}(A)(\vect{v},\vect{w})=0\;,\\
	\mbox{For $(n,m)=(5,8):$} & \hspace{1.5cm}\mathcal R_j \;\B{5}{8}(A)(\vect{v},\vect{w})=0\;, \;j=1,2,3
	\end{align*}
	with $\mathcal L$, $\mathcal R_j$, $j=1,2,3$, defined in Eqs. (\ref{Operator F4}), (\ref{Operators F8}) and where  we think of $\B{n}{m}(A)(\vect{v},\vect{w})$ as a function of $\vect{v}$ for $\vect{w}$ fixed.
\end{proposition}
\begin{proof}
	The case $(n,m)=(3,4)$. Let $g(\theta)=\mathrm{exp} ({\imath\theta})$ and $\T(g(\theta))$ given in Eq. (\ref{z'=T(g)z m=4}). From Corollary \ref{extended berezin symbol invariant one variable}
	\begin{equation}\label{kernel B34}
	\B{3}{4}(A)\big(\T(g(\theta))\vect{v},\vect{w}\big)= \B{3}{4}(A)(\vect{v},\vect{w}).
	\end{equation}

	By considering the partial derivative  of both sides in Eq. (\ref{kernel B34}) with respect to $\theta$ and evaluating the resulting equation at the point $\theta=0$, we obtain that $\B{3}{4}(A)(\vect{z},\vect{w})$ must belong to the kernel of the operator $\mathcal L$.
	
	The case $(n,m)=(5,8)$. Let $g(\theta,\alpha,\gamma) \in \mathrm{SU}(2)$ and $\T(g(\theta,\alpha,\gamma))$ given in Eq. (\ref{z'=T(g)z m=8}). From Corollary \ref{extended berezin symbol invariant one variable}
	\begin{equation}\label{kernel B58}
	\B{5}{8}(A)\big(\T(g(\TAG))\vect{v},\vect{w}\big)= \B{5}{8}(A)(\vect{v},\vect{w}).
	\end{equation}
	
	We consider the expression for $g(\TAG) \in \mathrm{SU}(2)$ given in Eq. (\ref{g parametrizada}). In a similar way that for the case $n=3$, we can prove that $\B{5}{8}(\vect{z},\vect{w})$ is in the kernel of the operators $\mathcal R_1$, $\mathcal R_2$ and $\mathcal R_3$ by considering the partial derivatives of both sides in Eq. (\ref{kernel B58}) with respect to $\alpha,\theta,\gamma$, respectively, and then evaluating at the point $(\TAG)=(0,0,0)$ (we actually need to take the limit $\theta \to 0$ in the last case).

\end{proof}

\section{The star product}\label{section star product}

In Ref. \cite{B-74}, Berezin showed that the product (\ref{Eq 3}) will allow us to define a star product, which will be denoted by $\spn{n}{m}$, on  the algebra $\A{n}{m}$ which consists of Berezin symbols for bounded linear operators with domain in $L^2(\Sn{n})$. See Ref. \cite{BF-78} for the standard definition of star product. Thus, we have the following %(see Refs. \cite{BF78} for the standard definition of star-product), i.e.
\begin{definition}
	For $f_1,f_2 \in \A{n}{m}$,
	\begin{align}
	\big(f_1 \spn{n}{m} f_2\big)(\vecz) & =\int_{\mathbb C^m} f_1(\vecz,\vecu) f_2(\vecu,\vecz) \frac{|\Q{m}(\vecu,\vecz)|^2}{\Q{m}(\vecz,\vecz)}\du{m}(\vect u),\label{expression star product}
	\end{align}
	where the functions $f_j(\vect v,\vect w)$, $j=1,2$, are the analytic continuation of $f_j(\vect v)$  to  $\mathbb C^m\times \mathbb C^m$ (see Eq. (\ref{extended covariant symbol})).
\end{definition}

\begin{remark}
	From Eqs. (\ref{extended berezin symbol invariant}) and (\ref{csf invariant second variable}), it follows  that the above star product is $\mathcal G_m$-invariant in the sense that
	\begin{equation}
	(f_1\circ \T(g))\spn{n}{m} (f_2\circ \T(g))= (f_1\spn{n}{m} f_2)\circ \T(g)\;, \hspace{0.5cm} \forall g \in \mathcal G_m, \;\forall f_1,f_2 \in \A{n}{m},
	\end{equation}
	where $\T(g)$ is given by the action of $\mathcal G_m$ on $\mathbb C^m$ (see Eqs. (\ref{z'=T(g)z m=2}), (\ref{z'=T(g)z m=4}) and (\ref{z'=T(g)z m=8})).
	Even more, in the next section we will prove that the star product $\spn{n}{m}$ is $\mathfrak F_m$-invariant, where $\mathfrak F_2=\mathrm{SU}(2)$, $\mathfrak F_4= \mathrm{SU}(2)\times \mathrm{SU}(2)$ and $\mathfrak F_8=\mathrm{SU}(4)$.
\end{remark}

In this section we verify that this noncommutative star product $\spn{n}{m}$ satisfies the usual requirement on the semiclassical limit, i.e. as $\hbar \to 0$
\begin{equation*}
f_1 \spn{n}{m} f_2(\vecz)=f_1(\vecz)f_2(\vecz)+\hbar  \tilde{\mathcal B}(f_1,f_2)(\vecz)+\mathrm O(\hbar^2),\hspace{0.5cm}\vecz\in \tilde{\mathbb C}^m, f_1,f_2\in \A{n}{m}, 
\end{equation*}
where $\tilde{\mathcal B}(\cdot,\cdot)$ is a certain bidifferential operator of the first order.

\begin{theorem} \label{theorem star product}Let $(n,m)=(2,2),(3,4),(5,8)$. The star product $\spn{n}{m}$ (see Eq. (\ref{expression star product})) satisfies
	\begin{enumerate}
		\item[a) ] $f \spn{n}{m} 1=1\spn{n}{m} f= f$, for all $f \in \A{n}{m}$,
		\item[b) ] $\spn{n}{m}$ is associative, and
		\item[c) ] for $f_1, f_2 \in \A{n}{m}$ and $\vect z \in \tilde{\mathbb C}^m$, we have the following asymptotic expression when $\hbar \to 0$
		\begin{align}
		f_1 \spn{n}{m} f_2(\vecz) &= f_1(\vecz) f_2(\vecz) + \hbar\left[\sum_{\ell=1}^{m}\partial_{u_\ell} f_2(\vecu,\vecz)\partial_{\overline u_\ell} f_1(\vecz,\vecu)\right]_{\vecu=\vecz} + \mathrm O(\hbar^2), \label{asymptotic star product}
		\end{align}
		where the functions $f_j(\vecz,\vecu)$, $j=1,2$, are the analytic continuation of $f_j(\vecz)$  to  $\mathbb C^m\times \mathbb C^m$ (see Eq. (\ref{extended covariant symbol})).
	\end{enumerate}
\end{theorem}

\begin{proof}
	a) Let $f \in \A{n}{m}$, then $f=\B{n}{m}(A)$ with $A\in \mathbf B(L^2(\Sn{n}))$. From  Eqs. (\ref{inner product cs = reproducing kernel}), (\ref{extended covariant symbol}) and Proposition \ref{complete} we have
	\begin{align*}
	f \spn{n}{m} 1 (\vect z)  & = \int_{\mathbb C^m} \big\langle A
	\csw,\csz\big\rangle_{\Sn{n}} \frac{\big\langle \csz,\csw\big\rangle _{\Sn{n}}} {\big\langle \csz,\csz\big\rangle_{\Sn{n}}} \du{m}(\vect w) =  f(\vect z).
	\end{align*}
	Analogously, $1 \spn{n}{m} f=f$.
	
	b) The associativity follows from the fact that the composition in the algebra of all bounded linear operator on $L^2(\Sn{n})$ is associative.
	
	%Let
	%\begin{equation}
	%\mathbf I(\vecz):= \frac{1}{(\pi\hbar)^m}\int_{\mathbb C^n} f_\vecz(\vecu,\vecu)\frac{|\Q{m}(\vecw,\vecz)|^2}{||\Q{m}(\vecz,\vecz)||^2}e^{|\vecw|^2/\hbar}\mathrm d \vecw \mathrm d \overline \vecw\;,
	%\end{equation}
	%where $f_\vecz(\vecu,\overline\vecu)=f_1(\vecu,\vecz) f_2(\vecz,\vecu)$.
	%We will use the stationary phase method to obtain Eq. (\ref{asymptotic star product}). In order to applied this theorem, 

	c) Let us first assume that $\vecz \ne 0$. The case $\vecz=0$ can be easily studied, see below.
	
	Case $(n,m)=(2,2)$:
	From Eq. (\ref{integral expresion Q_m})
	\begin{align}
	\frac{|\Q{2}(\vecu,\vecz)|^2}{\Q{2}(\vecz,\vecz)}\e^{-|\vecu|^2/\hbar} & = \frac{\exp(-|\vect{u}-\vect{z}|^2/\hbar)} {2(1+\exp(-2|\vect{z}|^2/\hbar))} +
	\frac{\exp(-|\vect{u}+\vect{z}|^2/\hbar)}{2(1+\exp(-2|\vect{z}|^2/\hbar))}\nonumber \\
	&\hspace{0.5cm} 
	+ \frac{\cos(2\mathrm{Im}(\vect{u}\cdot \vect{z})/\hbar)\exp(-|\vect{u}|^2/\hbar)\exp(-|\vect{z}|^2/\hbar)}{1+\exp(-2|\vect{z}|^2/\hbar)}.\label{principal m=2}
	\end{align}
	Since the last term in Eq. (\ref{principal m=2}) is $\mathrm O(\hbar^\infty)$, where $\mathrm O(\hbar^\infty )$ denotes a quantity tending to zero faster than any power of $\hbar$, we have from Eqs. (\ref{expression star product}) and (\ref{du})
	\begin{equation}\label{m=2 help}
	(f_1 \spn{2}{2} f_2)(\vecz)=\frac{\mathbf I(\vect{z}) + \mathbf I(-\vect{z})}{2(1+\exp(-2|\vect{z}|^2/\hbar))} + \mathrm O(\hbar^\infty)%=\frac{1}{2}\left(\mathbf I(\vect{z}) + \mathbf I(-\vect{z})+\mathrm O(\hbar^\infty)\right)\;
	\end{equation}
	where 
	\begin{equation}\label{I(v) integral expression}
	\mathbf I(\vect{v}):= \frac{1}{(\pi\hbar)^2} \int \beta_{\vect{z}}(\vect{u},\overline{\vect{u}})
	\mathrm{exp}\left(-\frac{1}{\hbar}|\vect u - \vect v|^2\right)
	%\mathrm{exp}(-|\vect{u}-\vect{v}|/\hbar) 
	\mathrm d\vect{u} \mathrm d \overline{\vect{u}}
	\end{equation}
	with $\beta_{\vect{z}}(\vect u,\overline{ \vect u})= f_1(\vect{z},\vect{u}) f_2(\vect{u},\vect{z})$. Notice however that $\mathbf I(\vect v)$ is just the standard formula for the solution at time  $t=\hbar/4$ of the heat equation on $\mathbb C^2=\mathbb R^{4}$ with initial data $\beta_\vect z$; i.e.
	\begin{equation}\label{formula solution heat equation}
	\mathbf I(\vect v)= \left.\sum_{\ell=0}^{\infty} \hbar^\ell \big(\partial_{\vect u\overline{\vect u}}\big )^\ell \beta_\vect z \right|_{\vect u=\vect v},
	\end{equation}
	where $\partial_{\vecu\overline \vecu}=\sum_j\partial_{u_j\overline u_j}$  denotes the Laplace operator. Eq. (\ref{formula solution heat equation}) can also be obtained using the stationary phase method (see Eq. (\ref{stationary phase method})), however we do not include its description in order not to make this proof too long.
	
	Thus, from Eqs. (\ref{m=2 help}), (\ref{I(v) integral expression}) and (\ref{formula solution heat equation})
	\begin{equation}\label{n=2 star productg invariance}
	(f_1 \spn{2}{2} f_2) (\vect z)= \big(1+\hbar \partial_{\vect u \overline{\vect u}}\big) \left[\left.\frac{1}{2}\beta_\vect z(\vect u,\overline{\vect u}) \right|_{\vect u=\vect z} + \left.\frac{1}{2}\beta_\vect z(\vect u,\overline{\vect u}) \right|_{\vect u=-\vect z}\right] + \mathrm{O}(\hbar^2).
	\end{equation}
	where we have used $1+ \mathrm{exp}(-2|\vect z|^2/\hbar)=1+\mathrm O(\hbar^\infty)$.
	
	Note that for $\vect v \in \mathbb C^2$, the function $\beta_\vect z(\vect v,\overline{\vect v})= f_1(\vect z, \vect v) f_2(\vect v, \vect z)$ satisfies $\beta_\vect z(\vect v,\overline{\vect v}) = \beta_{\vect z}(-\vect v,-\overline{\vect v}) $ because $f_1,f_2 \in \A{2}{2} $ and Corollary \ref{extended berezin symbol invariant one variable}. 
	
	%Thus, from the chain rule, Eq. (\ref{n=2 star productg invariance}) and used  that the extended Berezin symbol is holomorphic in the first factor and anti-holomorphic in the second we obtain Eq. (\ref{asymptotic star product}).

	Thus, from the chain rule and Eq. (\ref{n=2 star productg invariance})
	\begin{align*}
	(f_1 \spn{2}{2} f_2)(\vecz) &= f_1(\vecz) f_2(\vecz) + \hbar\sum_{\ell=1}^{2}\Big[\partial_{u_\ell} f_2(\vecu,\vecz)\partial_{\overline u_\ell} f_1(\vecz,\vecu)\Big]_{\vecu=\vecz} + \mathrm O(\hbar^2),
	\end{align*}
	where we have used that the extended Berezin symbol is holomorphic in the first factor and anti-holomorphic in the second.\\

	Case $(n,m)=(3,4)$: From Eqs. (\ref{expression star product}), (\ref{integral expresion Q_m}) and (\ref{du}) we have
	\begin{align}
	\big(f_1 \spn{3}{4} f_2\big)(\vecz) & = \frac{\e^{|\vecz|^2/\hbar}}{\Q{4}(\vecz,\vecz)} \frac{1}{4\pi^6\hbar^4}%\nonumber\\
	%& \hspace{1cm}
	\int\limits_{\psi=0}^{2\pi}  \int\limits_{\vecw \in \mathbb C^4} \int\limits_{\theta=0}^{2\pi} \beta_\vecz(\vecu,\overline \vecu) \mathrm{exp} \left(\frac{\imath}{\hbar}  \p_\psi(\vecu, \overline \vecu,\theta)\right) \mathrm d\theta
	\mathrm d \vecu \mathrm d \overline \vecu \mathrm d \psi, \label{SB}
	\end{align}
	where $\beta_\vecz(\vecu,\overline \vecu)=f_1(\vecz,\vecu)f_2(\vecu,\vecz)$ and the phase function $\p_\psi$ is 
	\begin{equation*}
	\p_\psi(\vecu,\overline\vecu,\theta)=\imath\left(|\vecu|^2 + |\vecz|^2-\vecu \cdot \T(g(\psi))\vecz - \overline{\vecu\cdot 
		\T(g(\theta))\vecz}\right)\;,
	\end{equation*}
	with $g(\psi)=\mathrm{exp}(\imath \psi)$ and $g(\theta)=\mathrm{exp}(\imath \theta)$. To obtain the asymptotic expansion (\ref{asymptotic star product}), we can use the stationary phase method with complex phase, see Eq. (\ref{stationary phase method}), in the integral appearing on the right hand side of Eq. (\ref{SB}).
	
	For our purpose, we need to consider the gradient and Hessian matrix of the function $\p_\psi$ with respect to the nine variables $\theta$, $x_j=\Re(u_j)$, $y_j=\Im(u_j)$, $j=1,2,3,4$. It is actually more convenient to consider the derivatives of $\p_\psi$ with respect to the variables $\theta$, $u_j$, $\overline u_j$ , $j = 1, 2, 3, 4$. Namely, 
	\begin{enumerate}
		\item[i) ]  the condition $\nabla_{\vecx,\vecy,\theta} \p_\psi=0$ is equivalent to $\nabla_{\vecu,\overline{\vecu},\theta} \p_\psi = 0$ with $\vecx = (x_1,\ldots, x_4 )$, $\vecy = (y_1,\ldots, y_4)$, $\vecu =(u_1,\ldots,u_4)$, and 
		\item[ii) ] to obtain the Hessian matrix of $\p_\psi$ with respect to the variables $\vecx, \vecy, \theta$ we use the following equalities: $\partial_{x_jx_k}=\partial_{u_ju_k}+\partial_{u_j\overline u_k}+\partial_{\overline u_ju_k}+\partial_{\overline u_j\overline u_k}$, $\partial_{x_jy_k} = \imath(\partial_{u_ju_k}-\partial_{u_j\overline u_k}+\partial_{\overline u_ju_k}-\partial_{\overline u_j\overline u_k})$, $\partial_{y_jy_k}=-\partial_{u_ju_k} +\partial_{u_j\overline u_k} +\partial_{\overline u_ju_k}-\partial_{\overline u_j\overline u_k}$, $\partial_{x_j\theta}=\partial_{u_j\theta}+\partial_{\overline u_j\theta}$, $\partial_{y_j\theta} =\imath(\partial_{u_j\theta} -\partial_{\overline u_j\theta})$.
	\end{enumerate}
	
	Notice that $\Im \p_\psi \ge 0$ because for all $\psi$ and $\theta$
	\begin{equation*}
	|\Re(z_j \overline u_j e^{\imath\psi} + \overline z_j u_j e^{-\imath\theta})| \le 2|z_j|\;|u_j|\;\;\mbox{with } j=1,2,3,4.
	\end{equation*}
	
	Given $\psi$ fixed, $\psi \in (0,2\pi)$, the gradient $\nabla_{\vect u,\overline{\vect u},\theta}\p_\psi$ is zero if and only if $\theta=\psi$, and $\vecu=\T(g(\psi)\vecz)$ (i.e. $u_j=z_j \mathrm{exp}(-\imath\psi)$, $j=1,2$ and $u_j=z_j \mathrm{exp}(\imath \psi)$, $j=3,4$), where we have used that $\vect z \in \tilde{\mathbb C}^m$, i.e. $\vect z$ satisfies the condition $|z_1|^2+|z_2|^2 =|z_3|^2+|z_4|^2$ (see Eq. (\ref{condition (3,4)})). Moreover, the Hessian matrix of $\p_\psi$ evaluated at the critical point $\vecx_0+\imath \vecy_0=\vecu_0=\T(g(\psi)\vecz)$, $\theta_0=\psi$ is equal to
	\begin{equation}\label{Hessian matrix (n,m)=(3,4)}
	\mathrm A:=\p''_\psi(\vecu_0,\overline\vecu_0,\theta_0)=\left(
	\begin{array}{c c c}
	2\imath \mathbf I_4 & \mathbf 0_4 &  \MB{4}\vecu_0\\[0.1cm]
	\mathbf 0_4 & 2\imath \mathbf I_4 & -\imath \MB{4} \vecu_0\\[0.1cm]
	(\MB{4}\vecu_0)^t & -\imath (\MB{4}\vecu_0)^t & \imath |\vecz|^2
	\end{array}
	\right)
	\end{equation}
	where $\mathbf I_s$ and $\mathbf 0_s$ denote the identity matrix and zero matrix of size $s$ respectively, $A^t$ denotes the transpose matrix of a given matrix $A$ and  
	\begin{equation*}
	\MB{2\ell}=\left(\begin{array}{c c}
	-\mathbf I_\ell & \mathbf 0_\ell\\
	\mathbf 0_\ell & \mathbf I_\ell
	\end{array}
	\right),\quad\ell \in \mathbb N.
	\end{equation*}

	Then from Eqs. (\ref{Hessian matrix (n,m)=(3,4)}) and (\ref{determinant}), with $\mathcal D=\imath |\vecz|^2$, we have
	\begin{equation*}
	\det (\p''_\psi(\vecu_0,\overline \vecu_0,\theta_0))=(2\imath)^8 \imath |\vecz|^2.
	\end{equation*}
	From the stationary phase method we obtain that
	\begin{equation}\label{tfe}
	\big(f_1 \spn{3}{4} f_2\big)(\vecz) = \frac{\e^{|\vecz|^2/\hbar}}{\Q{4}(\vecz,\vecz)} \frac{\sqrt{2\hbar}}{2^2\pi\sqrt\pi |\vecz|}
	\int_0^{2\pi} \left[\sum_{\ell<k} \hbar^\ell \mathbf M_\ell \beta_\vecz\Big|_{cp} + \mathrm O(\hbar^k)\right] \mathrm d \psi
	\end{equation}
	where
	\begin{equation}\label{Mj}
	\mathbf M_\ell \beta_\vecz\Big|_{cp}= \sum_{s=\ell}^{3\ell} \frac{\imath^{-\ell} 2^{-s}}{ s!(s-\ell)!} \left[ \left(-\mathrm A^{-1}\right) \hat D \cdot \hat D\right]^s \beta_\vecz (\p_{cp})^{s-\ell}\biggl|_{cp}
	\end{equation}
	with $\mathfrak g\big|_{cp}$ denoting the evaluation  at the critical point $\vect u_0, \theta_0$ of a given function $\mathfrak g$, 
	\begin{equation}\label{definition f|cp n=3}
	\p_{cp}=\p_{cp}(\vecu,\overline \vecu,\theta) = \imath \left(-\T(g(\theta))\vecz+\vecu_0+\imath(\theta-\psi)\MB{4}\vecu_0 \right) \cdot \vecu-\frac{\imath}{2}|\vecz|^2(\theta-\psi)^2,
	\end{equation}
	and $\hat D$ the column vector of size 9 whose entries are defined by: $(\hat D)_{j}=\partial_{x_j}$, $(\hat D)_{4+j}=\partial_{y_j}$, $j=1,\ldots,4$, and $(\hat D)_9=\partial_\theta$. The last Eq. (\ref{definition f|cp n=3}) is a consequence of equalities $|\vecu_0|=|\vecz|$ and $\p_\psi(\vect u_0,\overline{\vect u}_0,\theta_0)=0$.%that $\p_\psi=0$ at the critical point.

	In order to estimate $\mathbf M_1 \beta_\vecz\big|_{cp}$,
	%$\mathbf M_1 \beta_\vecz(\vecu_0, \overline{\vecu}_0,\theta_0)$, 
	we first need to obtain the inverse of the matrix $A$. From Eqs. (\ref{Hessian matrix (n,m)=(3,4)}) and (\ref{inverse}), with $\mathcal D=\imath |\vecz|^2$,
	\begin{equation*}
	\mathrm A^{-1}=\left(\begin{array}{c c c}
	\displaystyle \frac{1}{2\imath} \mathbf I_4-\frac{1}{4\imath |\vecz|^2}\MB{4} \vecu_0 \vecu_0^t \MB{4} & \displaystyle \frac{1}{4|\vecz|^2}\MB{4} \vecu_0 \vecu_0^t \MB{4} & \displaystyle \frac{1}{2|\vecz|^2}\MB{4} \vecu_0 \\[0.5cm]
	\displaystyle \frac{1}{4|\vecz|^2}\MB{4} \vecu_0 \vecu_0^t \MB{4} &  \displaystyle \frac{1}{2\imath} \mathbf I_4+\frac{1}{4\imath |\vecz|^2}\MB{4} \vecu_0 \vecu_0^t \MB{4} & \displaystyle\frac{1}{2\imath |\vecz|^2}\MB{4} \vecu_0\\[0.5cm]
	\displaystyle \frac{1}{2|\vecz|^2}\vecu_0^t \MB{4} & \displaystyle \frac{1}{2\imath|\vecz|^2}\vecu_0^t \MB{4} & \displaystyle \frac{1}{\imath |\vecz|^2}
	\end{array}
	\right).
	\end{equation*}
	
	Using the following equalities: $\MB{4}\vecu_0\vecu_0^t\MB{4} \mathbf p \cdot \overline{\mathbf p}=((\MB{4}\vecu_0)^t \mathbf p)^2$ for all $\mathbf p\in \mathbb C^4$ and $\partial_{u_j}=(\partial_{x_j}- \imath \partial_{y_j})/2$, and easy linear algebra manipulations, we can show
	\begin{equation} \label{operator in M-3}
	-(\mathrm A^{-1}) \hat D\cdot \hat D = 
	%\left(-\mathbf A^{-1}\right)_{jk} \hat D_j\hat D_k= 
	2\imath \partial_{\vecu\overline \vecu} - \frac{\imath}{|\vecz|^2}\left[(\MB{4} \vecu_0)^t \partial_\vecu-\imath \partial_\theta\right]^2,
	\end{equation}
	where $\partial_{\vecu\overline \vecu}=\sum_{j=1}^4\partial_{u_j\overline u_j}$  and $\partial_\vecu$ denote the Laplace operator and the column vector of size 4 whose $j$ entry is $\partial_{u_j}$ (i.e $(\partial_{\vecu})_j= \partial_{u_j}$) respectively. 
	
	From Eqs. (\ref{Mj}) and (\ref{operator in M-3})
	\begin{align}
	%\mathbf M_0\beta_\vecz (\vecu_0,\overline{\vecu_0},\theta_0) 
	\mathbf M_0\beta_\vecz \Big|_{cp} 
	& =\beta_\vecz(\vecu_0,\overline\vecu_0), \label{M0 n=3}\\
	%\mathbf M_1 \beta_\vecz(\vecu_0,\overline{\vecu_0},\theta_0)
	\mathbf M_1 \beta_\vecz\Big|_{cp}
	& =\left[\partial_{\vecu\overline\vecu} - \frac{1}{2|\vecz|^2} \left((\MB{4} \vecu_0)^t \partial_\vecu \right)^2-\frac{1}{2|\vecz|^2} \vecu_0^t \partial_\vecu%\right.\nonumber\\
%	&\hspace{1cm}\left. 
+\frac{1}{2^3|\vecz|^2}\right] \beta_\vecz(\vecu_0,\overline\vecu_0).\label{M1 n=3}
	\end{align}
	where we have used the equality $\MB{4}\vect u_0 \cdot \vect u_0=0$ to obtain Eq. (\ref{M1 n=3}). 
	
	Notice that the right side of Eqs. (\ref{M0 n=3}) and (\ref{M1 n=3}) still depend on the variable $\psi$ because $\vecu_0 =\T(g(\psi))\vecz$.
	
	Thus,  from Eqs. (\ref{tfe}), (\ref{M0 n=3}) and (\ref{M1 n=3})
	\begin{align}
	(f_1\spn{3}{4} f_2)(\vecz) & = \frac{\e^{|\vecz|^2/\hbar}}{\Q{4}(\vecz,\vecz)} \frac{\sqrt{2\hbar}}{2^2\pi\sqrt\pi |\vecz|}\left\{
	\int_0^{2\pi} \left[1+\hbar \left(\partial_{\vecu\overline\vecu} - \frac{1}{2|\vecz|^2} \left((\MB{4} \vecu_0)^t \partial_\vecu\right)^2\right. \right. \right.\nonumber\\[.2cm]
	& \hspace{1cm}\left.\left.\left.-\frac{1}{2|\vecz|^2} \vecu_0^t \partial_\vecu +\frac{1}{2^3|\vecz|^2}\right)\right] \beta_\vecz(\vecu_0,\overline\vecu_0)\mathrm d \psi + \mathrm O(\hbar^2) \right\}.\label{SB1}
	\end{align}

	Note that for $\vecv \in \mathbb C^4$ and $g(\theta) \in S^1$, the function $\beta_\vecz(\vecv,\overline \vecv)=f_1(\vecz,\vecv)f_2(\vecv,\vecz)$ satisfies 
	\begin{equation*}
	\beta_\vecz(\vecv,\overline\vecv)=\beta_\vecz(\T(g(\theta))\vecv,\overline{\T(g(\theta))\vecv})
	\end{equation*}
	because %where the last equality is a consequence that 
	$f_1,f_2 \in \A{3}{4}$ and Corollary \ref{extended berezin symbol invariant one variable}. Then, from the chain rule
	\begin{align}
	\vecv^t \partial_\vecv \beta_\vecz(\vecv,\overline \vecv)\Big|_{\substack{\vecv=\vecz\\ \theta=\psi}} & =  \vecu_0^t \partial_\vecu \beta_\vecz(\vecu_0,\overline\vecu_0), \nonumber\\
	\partial_{\vecv \overline \vecv} \beta_\vecz (\vecv,\overline\vecv)\Big|_{\substack{\vecv=\vecz\\ \theta=\psi}} & = \partial_{\vecu\overline \vecu} \beta_\vecz(\vecu_0,\overline \vecu_0), \label{derivate invariance}\\
	\left[\big(\vecv^t \MB{4} \partial_\vecv\big)^2 -\vecv^t \partial_\vecv\right] \beta_\vecz (\vecv,\overline \vecv)\Big|_{\substack{\vecv=\vecz\\ \theta=\psi}} & =\biggl(\sum_{j,k=1}^4 (\MB{4})_{j,j}(\MB{4})_{k,k}v_j v_k\partial_{v_j} \partial_{v_k}\biggl) \beta_\vecz(\vecv ,\overline{\vecv})\biggl|_{\substack{\vecv=\vecz\\ \theta=\psi}}\nonumber \\
	& = \left((\MB{4} \vecu_0)^t \partial_\vecu\right)^2 \beta_\vecz (\vecu_0,\overline \vecu_0). \nonumber
	%\vecz^t B \partial_\vecz \beta_\vecz(\vecz ,\overline{\vecz}) & = \overline \vecz^t B \partial_{\overline \vecz} \beta_\vecz(\vecz,\overline{\vecz}) \label{derivada invariantes angulo}
	\end{align}

	Thus, from Eqs. (\ref{SB1}), (\ref{derivate invariance}), Proposition \ref{kernel F_m = Bessel},
	the asymptotic  expression of the modified Bessel function $\mathbf I_\nu$ (see Eq. (\ref{asymptotic expression modified Bessel function})) and the relation $\sqrt 2 |\rho_{(3,4)} (\vect z)|=|\vect z|^2$ we have
	\begin{align*}
	(f_1\spn{3}{4} f_2)(\vecz) & = \left.\left[1+ \hbar\left(\partial_{\vecu\overline\vecu}-\frac{1}{2|\vecz|^2}(\vecu^t \MB{4} \partial_\vecu)^2 \right)\right]\beta_\vecz(\vecu,\overline{\vecu})\right|_{\vecu=\vecz}+\mathrm O(\hbar^2),
	\end{align*}
	
	Using the equality $\vect{u}^t \MB{4} \partial_{\vect{u}} f_2(\vect{u},\vect{z})=0$ (see Proposition \ref{extended Berezin symbol in kernel}) and that the extended Berezin symbol is holomorphic in the first factor and anti-holomorphic in the second  we finally obtain Eq. (\ref{asymptotic star product}).
	%From Proposition \ref{extended Berezin symbol in kernel}: $\vect{u}^t \MB{4} \partial_{\vect{u}} f_2(\vect{u},\vect{z})=0$. Therefore, 
	%	\begin{align*}
	%	(f_1 \spn{3}{4} f_2)(\vecz) &= f_1(\vecz) f_2(\vecz) + \hbar\sum_{\ell=1}^{4}\Big[\partial_{u_\ell} f_2(\vecu,\vecz)\partial_{\overline u_\ell} f_1(\vecz,\vecu)\Big]_{\vecu=\vecz} + \mathrm O(\hbar^2),
	%	\end{align*}
	%where we have used that the extended Berezin symbol is holomorphic in the first factor and anti-holomorphic in the second.\\

	Case $(n,m)=(5,8)$: This case is similar to the case when $(n,m)=(3,4)$ but the computations are more involved. First note that from Eqs. (\ref{expression star product}), (\ref{integral expresion Q_m}) and (\ref{du})
	\begin{align}
	\big(f_1 \spn{5}{8} f_2\big)(\vecz) & = \frac{\e^{|\vecz|^2/\hbar}}{\Q{8}(\vecz,\vecz)} \frac{1}{(\pi\hbar)^8}\int\limits_{\theta,\alpha,\gamma}  \int\limits_{\vecu \in \mathbb C^8} \int\limits_{\tilde \theta,\tilde \alpha,\tilde \gamma} \mathrm{exp} \left(\frac{\imath}{\hbar} \p_{\TAG} (\vecu, \overline \vecu,\TTAG)\right)%\nonumber\\
	%& \hspace{1cm}\times 
	\beta_\vecz(\vecu,\overline \vecu)  \mathrm dm(\TTAG)\mathrm d \vecu \mathrm d \overline \vecu \mathrm d m(\TAG), \label{SB n=5}
	\end{align}
	with $\beta_\vecz(\vecu,\overline \vecu)=f_1(\vecz,\vecu)f_2(\vecu,\vecz)$ and the phase function $\p_{\TAG}$ is given by 
	\begin{equation*}
	\p_{\TAG}(\vecu,\overline\vecu,\TTAG)=\imath\left(|\vecu|^2 + |\vecz|^2-\vect u \cdot \T(g(\TAG))\vect z - \overline{\vecu\cdot \T(g(\TTAG))\vecz}\right),
	\end{equation*}
	where we are considering the expression for $g(\TAG), g(\TTAG)\in \mathrm{SU}(2)$ given in Eq. (\ref{g parametrizada}). For $\TAG$ fixed, $\theta \in (0,\pi/2)$ and $\alpha,\gamma \in (-\pi,\pi)$, the equations $\partial_{u_\ell} \p_{\TAG}=0$ and $\partial_{\overline u_\ell} \p_{\TAG}=0$, $\ell=1,\ldots,8$ imply $\vect u= \T(g(\TAG))\vect z= \T(g(\TTAG))\vect z$, which in turn implies (see Eq. (\ref{z'=T(g)z m=8}))
	\begin{equation} \label{VgL=L}
	\mathbf V\left(g^{-1}(\TTAG)g(\TAG) \right)\mathbf L \vect z= \mathbf L \vect z,
	\end{equation}
	where $\mathbf V(g)$ and $\mathbf L$ are defined in Eq. (\ref{matrix L}). Since $\vect z\ne \mathbf 0$ then $\mathbf L \vect z\ne \mathbf 0$, therefore we obtain from Eq. (\ref{VgL=L}) that $g^{-1}(\TTAG)g(\TAG)$ must be the identity matrix, which in turn implies $(\TTAG)=(\TAG)$. 
	
	Even more, we claim that
	\begin{equation}\label{partial f evaluated cp n=5}
	\partial_{\tilde \vartheta} \p_{\TAG}\Big|_{cp}=0, \hspace{0.5cm} \mbox{for }\vartheta=\TAG,
	\end{equation}
	where $\partial_{\tilde \vartheta} \p_{\TAG}\Big|_{cp}$ denotes $\partial_{\tilde \vartheta} \p_{\TAG}$ evaluated at $\vect u= \T(g(\TAG))\vect z$ and $(\TTAG)=(\TAG)$. To show this fact, note that from the explicit expression of the function $\p_{\TAG}$ and Eq.  (\ref{z'=T(g)z m=8}), we have
	\begin{equation*}
	\partial_{\tilde \vartheta} \p_{\TAG} \Big|_{cp} = -\imath \mathbf V(g^{-1}(\TAG)) \partial_{\vartheta} \mathbf V(g(\TAG)) \mathbf L \vect{z} \cdot \mathbf L\vect{\vecz},\hspace{0.5cm} \vartheta= \TAG.
	\end{equation*}
	
	From the expression for $g=g(\TAG)$ given in Eq. (\ref{g parametrizada}) we find
	\begin{align*}
	\partial_{\tilde \theta} \p_{\TAG} \Big|_{cp} & = 2\Im\left(e^{\imath (\gamma-\alpha)}[z_7 \overline z_1 + z_5 \overline z_3 -z_6 \overline z_4-z_8\overline z_2]\right),\\
	\partial_{\tilde \alpha} \p_{\TAG} \Big|_{cp} & = \cos^2\theta[|z_1|^2+|z_2|^2+|z_3|^2+|z_4|^2-|z_5|^2-|z_6|^2-|z_7|^2-|z_8|^2]\\
	& \hspace{.5cm}+ 2\Re\left(\sin(\theta)\cos(\theta) e^{\imath(\gamma-\alpha)}[z_5\overline z_3-z_6 \overline z_4+z_7\overline z_1 -z_8\overline z_2 ]\right),\\
	\partial_{\tilde \gamma} \p_{\TAG} \Big|_{cp} & = \sin^2\theta[|z_5|^2+|z_6|^2+|z_7|^2+|z_8|^2-|z_1|^2-|z_2|^2-|z_3|^2-|z_4|^2]\\
	&\hspace{0.5cm}+2\Re\left(e^{\imath(\gamma-\alpha) }\sin(\theta)\cos(\theta)[z_5 \overline z_3-z_6\overline z_4+z_7 \overline z_1 -z_8\overline z_2]\right).
	\end{align*}
	
	Since $\vect{z} \in \tilde{\mathbb C}^8$ (see Eq. (\ref{condition (5,8)})), %satisfies the restriction (\ref{condition (5,8)}) 
	then Eqs. (\ref{partial f evaluated cp n=5}) hold. Thus, the critical point is  
	\begin{equation}\label{critical point n=5}
	\vect{u}_0= \T(g(\TAG))\vect z, \hspace{0.5cm} (\tilde \theta_0,\tilde \alpha_0,\tilde \beta_0)=(\TAG).
	\end{equation}
	
	One can also check that $\Im \p_{\TAG}\ge 0$ on the domain of $\p_{\TAG}$ and that $\p_{\TAG}=0$ at the critical point. Moreover, the Hessian matrix of $\p_{\TAG}$ with respect to the variables $\vect x,\vect y, \TTAG$ (with $\vect x= \Re \vect u$ and $\vect y=\Im \vect u$) evaluated at the critical point is equal to
	\begin{equation}\label{Hessian matrix (n,m)=(5,8)}
	\mathrm A:=%f''_\psi(\vecu_0,\overline\vecu_0,\theta_0)=
	\left(
	\begin{array}{c c c c c}
	2\imath \mathbf I_8 & \mathbf 0_8 &   -\imath \T_\theta(g) \vect z &  -\imath \T_\alpha(g) \vect z & -\imath \T_\beta(g)\vect z\\[0.1cm]
	\mathbf 0_8 & 2\imath \mathbf I_8 &    - \T_\theta(g) \vect z &  - \T_\alpha(g) \vect z & - \T_\beta(g)\vect z\\[0.1cm]
	-\imath \big(\T_\theta \vect z\big)^t & - \big(\T_\theta \vect z\big)^t & \imath |\vect z|^2 & 0 & 0\\[0.1cm]
	-\imath \big(\T_\alpha \vect z\big)^t & - \big(\T_\alpha \vect z\big)^t & 0 & \imath |\vect z|^2 \cos^2\theta & 0\\[0.1cm]
	-\imath \big(\T_\beta \vect z\big)^t & - \big(\T_\beta \vect z\big)^t & 0 & 0 & \imath |\vect z|^2\sin^2 \theta 
	\end{array}
	\right)
	\end{equation}
	where $\T_\vartheta(g) \vect z:=\partial_\vartheta \T(g(\TAG)) \vect z= \mathbf L^\dagger \partial_\vartheta \mathbf V(g(\TAG)) \mathbf L \vect z$, $\vartheta=\TAG$ (see Eq. (\ref{z'=T(g)z m=8})).
	
	From Eqs. (\ref{Hessian matrix (n,m)=(5,8)}) and (\ref{determinant}), with $\mathcal D$ the diagonal matrix
	\begin{equation}\label{matrix D n=5}
	\mathcal D=\mathrm{diag}(\imath |\vect z|^2, \imath |\vect z|^2 \cos^2\theta,  \imath |\vect z|^2 \sin^2\theta),
	\end{equation}  
	we have that $\det (A)=(2\imath)^{16} \imath^3 |\vect z|^6 \cos^2 \theta \sin^2\theta.$
	%\begin{equation*}
	%\det (A)=(2\imath)^{16} \imath^3 |\vect z|^6 \cos^2 \theta \sin^2\theta.
	%\end{equation*}
	
	Thus, from the stationary phase method
	\begin{align}
	\big(f_1 \spn{5}{8} f_2\big)(\vecz)  & = \frac{\e^{|\vecz|^2/\hbar}}{\Q{8}(\vect z,\vect z)}\frac{1}{(\pi\hbar)^8} \int\limits_{\TAG}
	\left(\frac{2^3(\pi \hbar)^{19}}{|\vect z|^6 \cos^2\theta\sin^2\theta} \right)^{\frac{1}{2}} \frac{1}{2\pi^2}
	%\nonumber\\
	%& \hspace{0.5cm}\times 
	\left[\sum_{\ell<k} \hbar^\ell \mathbf M_\ell \left(\cos \tilde{\theta} \sin\tilde{\theta}\beta_\vecz\right) \Big|_{cp}%(\vecu_0,\overline{\vecu}_0,\theta_0,\alpha_0,\beta_0)
	+ \mathrm O(\hbar^k)\right] \mathrm dm(\TAG) \label{tfe n=5}
	\end{align}
	where
	\begin{equation}\label{Mj n=5}
	\mathbf M_\ell\left(\cos\tilde{\theta} \sin\tilde{\theta} \beta_\vecz\right)\Big|_{cp}= \sum_{s=\ell}^{3\ell} \frac{\imath^{-\ell} 2^{-s}}{ s!(s-\ell)!} \left[ \left(-\mathrm A^{-1}\right) \hat D \cdot \hat D\right]^s \cos\tilde{\theta} \sin\tilde{\theta}\beta_\vecz (\p_{cp})^{s-\ell}\biggl|_{cp}
	\end{equation}
	with $\mathfrak g\big|_{cp}$ denoting the evaluation  at the critical point $\vect u_0, \theta_0,\alpha_0,\gamma_0$ of a given function $\mathfrak g$, 
	%with $(\cdot)\big|_{cp}$ denoting a function $(\cdot)$ evaluated at the critical point $\vect u_0, \theta_0,\alpha_0,\beta_0$, 
	\begin{align}
	\p_{cp} & =\p_{cp}(\vecu,\overline \vecu,\TTAG)\nonumber\\
	& = \imath \Big(\vect u_0-\T(g(\TTAG))\vect z+ (\tilde \theta - \theta)\T_\theta(g) \vect z 
	+ (\tilde \alpha-\alpha)\T_\alpha(g) \vect z %\nonumber\\
	%& \hspace{0.3cm}
	+(\tilde\gamma -\gamma)\T_\gamma(g)\vect z\Big) \cdot \vecu\nonumber\\
	&\hspace{0.5cm}-\frac{\imath}{2}|\vecz|^2\Big[(\tilde\theta-\theta)^2 +\cos^2\theta(\tilde{\alpha}-\alpha)^2%\nonumber\\
	%&\hspace{0.3cm}
	+\sin^2\theta(\tilde{\gamma}-\gamma)^2 
	\Big],\label{definition f|cp n=5}
	\end{align}
	and $\hat D$ the column vector of size 19
	%\begin{equation*}
	%\hat D=(\partial_{x_1},\ldots, \partial_{x_8}, \partial_{y_1}, \ldots,\partial_{y_8}, \partial_{\tilde{\theta}}, \partial_{\tilde{\alpha}}, \partial_{\tilde{\beta}})
	%\end{equation*}
	%
	whose entries are defined by: $(\hat D)_{j}=\partial_{x_j}$, $(\hat D)_{8+j}=\partial_{y_j}$, $j=1,\ldots,8$, $(\hat D)_{17}=\partial_{\tilde\theta}$, $(\hat D)_{18}=\partial_{\tilde\alpha}$ and $(\hat D)_{19}=\partial_{\tilde\gamma}$. The last Eq. (\ref{definition f|cp n=5}) is a consequence of equalities $|\vecu_0|=|\vecz|$ and $\T_\vartheta(g)\vect z \cdot \vect u_0=0$, $\vartheta=\TAG$.

	In a similar way as we did for the case $n=3$, we obtain  the inverse of the matrix $A$ (see Eq. (\ref{Hessian matrix (n,m)=(5,8)})) using Eq. (\ref{inverse}), with $\mathcal D$ the diagonal matrix given in Eq. (\ref{matrix D n=5}). By considering the explicit expression for the inverse matrix $A^{-1}$, using the equality $\T_{\vartheta} (g) \vect z \left(\T_{\vartheta}(g) \vect z\right)^t  \overline{ \mathbf p} \cdot {\mathbf p}=(\T_\vartheta(g)\vect z\cdot \mathbf p)^2$, $\vartheta=\TAG$, for all $\mathbf p\in \mathbb C^8$,  and easy linear algebra manipulations, we can show
	\begin{align} 
	-(\mathrm A^{-1}) \hat D\cdot \hat D & = 
	%\left(-\mathbf A^{-1}\right)_{jk} \hat D_j\hat D_k= 
	2\imath \partial_{\vecu\overline \vecu} + \frac{\imath}{|\vect z|^2}\left[\frac{1}{\cos^2\theta} \Big((\T_\alpha(g)\vect z )^t \partial_\vecu+ \partial_{\tilde\alpha} \Big)^2%\right.\nonumber\\
	%& \hspace{0.5cm}\left.
	+\frac{1}{\sin^2\theta} \left((\T_\gamma(g)\vect z)^t \partial_\vecu+ \partial_{\tilde\gamma}\right)^2+ \Big((\T_\theta(g)\vect z)^t \partial_\vecu+ \partial_{\tilde\theta}\Big)^2\right], \label{operator in M-5}
	\end{align}
	where $\partial_{\vecu\overline \vecu}=\sum_{j=1}^8\partial_{u_j\overline u_j}$  and $\partial_\vecu$ denote the Laplace operator and the column vector of size 8 whose $j$ entry is $\partial_{u_j}$ (i.e $(\partial_{\vecu})_j= \partial_{u_j}$) respectively.

	From Eqs. (\ref{Mj n=5}) and (\ref{operator in M-5})
	\begin{align}
	%\mathbf M_0\beta_\vecz (\vecu_0,\overline{\vecu_0},\theta_0) 
	\mathbf M_0  \left(\cos \tilde{\theta}\sin\tilde{\theta}\beta_\vecz \right) \Big|_{cp} 
	& =\cos\theta\sin\theta\beta_\vecz(\vecu_0,\overline\vecu_0), \label{M0 n=5}\\
	%\mathbf M_1 \beta_\vecz(\vecu_0,\overline{\vecu_0},\theta_0)
	\mathbf M_1 \left(\cos \tilde{\theta}\sin\tilde{\theta} \beta_\vecz\right)\Big|_{cp} & 	=\frac{ \cos\theta \sin\theta}{2|\vect z|^2} \bigg[2|\vect z|^2 \partial_{\vect u \overline{\vect u}} -(\T(g)\vect z)^t\partial_\vect u +  \left[(\T_\theta(g) \vect z)^t\partial_\vect u\right]^2 %\nonumber\\
	%& 
	+ \frac{\left[(\T_\alpha(g) \vect z)^t\partial_\vect u\right]^2}{\cos^2 \theta}+ \frac{\left[(\T_\gamma(g) \vect z)^t\partial_\vect u\right]^2}{\sin^2 \theta}\nonumber\\
	& \quad-  \frac{(\MB{8}\T_\alpha(g)\vect z)^t\partial_\vect u}{\imath \cos^2\theta} - \frac{(\MB{8}\T_\gamma(g)\vect z)^t\partial_\vect u}{\imath\sin^2\theta}% \nonumber\\
	%&
	    + \left(\frac{\cos\theta}{\sin\theta} -\frac{\sin\theta} {\cos\theta}\right)(\T_\theta(g)\vect z)^t \partial_\vect u- \frac{3}{4}
	\bigg] \beta_\vect z(\vect u_0,\overline{\vect u_0}).\label{M1 n=5}
	\end{align}
	Notice that the right side of Eqs. (\ref{M0 n=5}) and (\ref{M1 n=5}) still depend on the variables $\TAG$, because $\vecu_0 =\T(g(\TAG))\vecz$ (see Eq. (\ref{critical point n=5})).

	On the other hand, since $f_1,f_2 \in \B{5}{8}$, we obtain from Corollary \ref{extended berezin symbol invariant one variable} that, for $\vect v \in \mathbb C^8$ and $g(\TTAG)\in \mathrm{SU}(2)$, the function $\beta_\vect z(\vect v,\overline{\vect v})=f_1(\vect z,\vect v)f_2(\vect v,\vect z)$ satisfies
	%
	%for $\vect v \in \mathbb C^8$ and $g(\TTAG)\in \mathrm{SU}(2)$, the function $\beta_\vect z(\vect v,\overline{\vect v})=f_1(\vect z,\vect v)f_2(\vect v,\vect z)$ satisfies (see Corollary \ref{extended berezin symbol invariant one variable}):
	\begin{equation*}
	\beta_\vecz(\vecv,\overline\vecv)=\beta_\vecz(\T(g(\TTAG))\vecv,\overline{\T(g(\TTAG))\vecv}). %\quad \forall \vect v \in \mathbb C^8, \forall g(\TTAG)\in \mathrm{SU}(2) .
	\end{equation*}
	Then, from the chain rule
	\begin{align}
	\vecv^t \partial_\vecv \beta_\vecz(\vecv,\overline \vecv)\Big|_p
	%\Big|_{\substack{\vecv=\vecz\\ (\TTAG)=(\TAG)}} 
	& = \vecu_0^t \partial_\vecu \beta_\vecz(\vecu_0,\overline\vecu_0) %\nonumber\\
	%&
	 =  -\left(\T_{\theta \theta} (g)\vect z\right)^t \partial_\vect u \beta_\vect z(\vect u_0,\overline{\vect u_0})\nonumber\\
	\partial_{\vecv \overline \vecv} \beta_\vecz (\vecv,\overline\vecv) \Big|_p
	%\Big|_{\substack{\vecv=\vecz\\ (\TTAG)=(\TAG)}} 
	& = \partial_{\vecu\overline \vecu} \beta_\vecz(\vecu_0,\overline \vecu_0) \label{derivate invariance n=5}\\[0.2cm]
	\Big[-(\mathcal R_1)^2+(\mathcal R_2)^2  %\hspace{1.5cm}&  \\ 
	-(\mathcal R_3)^2 +2\vect v^t \partial_\vect v\Big]  \beta_\vect z(\vect v,\overline{\vect v}) \Big|_p
	%\Big|_{\substack{\vecv=\vecz\\ (\TTAG)=(\TAG)}} 
	& = \left( \left[(\T_\theta(g) \vect z)^t\partial_\vect u\right]^2 + \frac{\left[(\T_\alpha(g) \vect z)^t\partial_\vect u\right]^2}{\cos^2 \theta}%\right.\nonumber\\
	%&  \hspace{0.5cm}\left.
	+ \frac{\left[(\T_\gamma(g) \vect z)^t\partial_\vect u\right]^2}{\sin^2 \theta}\right) \beta_\vect z(\vect u_0,\overline{\vect u _0})\nonumber \\[0.2cm]
	-2\vect v^t \partial_\vect v \beta_\vect z(\vect v,\overline{\vect v}) \Big|_p
	%\Big|_{{\vecv=\vecz\\ (\TTAG)=(\TAG)}} 
	& = \left(-  \frac{(\MB{8}\T_\alpha(g)\vect z)^t\partial_\vect u}{\imath \cos^2\theta} - \frac{(\MB{8}\T_\gamma(g)\vect z)^t\partial_\vect u}{\imath\sin^2\theta}\right. \nonumber\\
	& \hspace{0.5cm}\left.+ \left(\frac{\cos\theta}{\sin\theta} -\frac{\sin\theta} {\cos\theta}\right)(\T_\theta(g)\vect z)^t \partial_\vect u\right) \beta_\vect z(\vect u_0,\overline{\vect u _0})\nonumber
	\end{align}
	where for a given function $\mathfrak g$ we denote by $\mathfrak g|_{p}$ the evaluation of $\mathfrak g$ at $\vect v=\vect z$, $(\TTAG)=(\TAG)$ and $\mathcal R_\ell$, $\ell=1,2,3$ are give in Eq. (\ref{Operators F8}).
	
	Thus, from Proposition \ref{kernel F_m = Bessel}, the asymptotic  expression of the modified Bessel function $\mathbf I_\nu$ (see Eq. (\ref{asymptotic expression modified Bessel function})), Eqs. (\ref{tfe n=5}), (\ref{M0 n=5}), (\ref{M1 n=5}), (\ref{derivate invariance n=5}),  and the relation $\sqrt 2 |\rho_{(5,8)} (\vect z)|=|\vect z|^2$ we have

	\begin{align*}
	(f_1*f_2)(\vect z)& = \left(1+ \hbar\left[\partial_{\vecu\overline\vecu} + \frac{1}{2|\vecz|^2} \Big(-(\mathcal R_1)^2+(\mathcal R_2)^2% \right.\right.\nonumber\\
	%& \hspace{0.5cm}\left.\left.
	-(\mathcal R_3)^2  \Big)\bigg]\right) \beta_\vecz(\vecu,\overline\vecu) \right|_{\vect u =\vect z} + \mathrm O(\hbar^2). %\label{M1 n=5}
	\end{align*}
	%where $\mathcal R_\ell$, $\ell=1,2,3$, are defined in Eq. (\ref{Operators F8}). 
	Then, from Proposition \ref{extended Berezin symbol in kernel} (applying to $f_2$) and using that the extended Berezin symbol is holomorphic in the first factor and anti-holomorphic in the second we obtain Eq. (\ref{asymptotic star product}).
	%\begin{align*}
	%(f_1 \spn{3}{4} f_2)(\vecz) &= f_1(\vecz) f_2(\vecz) + \hbar\sum_{\ell=1}^{8}\Big[\partial_{u_\ell} f_2(\vecu,\vecz)\partial_{\overline u_\ell} f_1(\vecz,\vecu)\Big]_{\vecu=\vecz} + \mathrm O(\hbar^2).
	%\end{align*}

	Finally, suppose we have the case $|\vect z|=0$. Then $\rho_{(n,m)}(\vect z)=0$ and therefore the coherent state $\Phi_{\rho_{(n,m)}(\vect z)}^{(\hbar)}$ is the constant function $1$ on the whole sphere with $L^2(\Sn{n})$ norm equal to one. Thus, in a similar way as we did for the case $n=2$ (see Eqs. (\ref{I(v) integral expression}) and (\ref{formula solution heat equation})) we conclude from the stationary phase method: 
	\begin{align*}
	f_1 \spn{n}{m} f_2(\vect z)& = \frac{1}{(\pi\hbar)^m}\int_{\mathbb C^m} f_1(\vect 0,\vect u)f_2(\vect u,\vect 0)\mathrm{exp}\left(- \frac{|\vect u|^2}{\hbar}\right) \mathrm d \vect u \mathrm d \overline{\vect u}\\
	& = f_1(\vect 0) f_2(\vect 0) + \hbar\sum_{\ell=1}^{m}\Big[\partial_{u_\ell} f_2(\vect u,\vect 0)\partial_{\overline u_\ell} f_1(\vect 0,\vect u)\Big]_{\vect u=\vect 0} + \mathrm O(\hbar^2).
	\end{align*}
\end{proof}

\section{Invariance of star-product $\spn{n}{m}$}\label{section invariance star product}

In this section we show that the star product $\spn{n}{m}$ defined on the algebra $\A{n}{m}$ is invariant under the group $\mathfrak F_m=\mathrm{SU}(2), \mathrm{SU}(2)\times \mathrm{SU}(2), \mathrm{SU}(4)$ for $m=2,4,8$ respectively, i.e. it satisfy	
\begin{equation*}
\big(f_1\circ \OL(g)\big) \spn{n}{m} \big(f_2 \circ \OL(g)\big)=(f_1 \spn{n}{m} f_2) \circ \OL(g),
%\OL(g)^*(f_1 \spn{n}{m} f_2)=\big(\OL(g)^*f_1 \big) \spn{n}{m} \big(\OL(g)^*f_2\big), 
\hspace{0.5cm} \forall g \in \mathfrak F_m,\; \forall f_1,f_2 \in \A{n}{m},
\end{equation*}
where $\OL(g)$ denotes the action of the group $\mathfrak F_m$ on $\mathbb C^m$ and it is defined by the following equation
%
%
%Donde consideramos que un producto estrella $*$ es $G$-invariante (G un grupo de Lie) si para cualquier $g \in G$ y todo $f_1,f_2 \in$ se tiene
%donde 
%Sea $G$ un grupo de Lie actuando sobre una variedad $M$, dicha acci\'on es $G \times M \to M$: $(g,p)\mapsto \sigma(g)p$. Un producto estrella se dice que es $G$-invariante, si para cualquier $g \in G$ y todo $\phi,\psi \in C^\infty(M)$  se tiene

\begin{equation}\label{action m=2,4,8}
\vect{z}'=\OL(g) \vect{z},\hspace{0.5cm} g \in \mathfrak F_m
\end{equation}
with $\OL(g)$ given by:

Case $m=2$: For $g \in \mathrm{SU}(2)$
\begin{equation*}
\OL(g)=g
\end{equation*}

Case $m=4$: For $g=(V,W)\in \mathrm{SU}(2)\times \mathrm{SU}(2)$
\begin{equation*}
\OL(g)=\begin{pmatrix}
V & \mathbf 0_2\\
\mathbf 0_2 & W 
\end{pmatrix},
\end{equation*}
where $\mathbf 0_\ell$ denotes the zero matrix of size $\ell$.

Case $m=8$: For $g\in \mathrm{SU}(4)$
\begin{equation} \label{action SU(4) C8}
\OL(g)=\begin{pmatrix}
U & \mathbf 0_4\\
\mathbf 0_4 & EUE 
\end{pmatrix},
\end{equation}
where $E$ is the following orthogonal matrix
\begin{equation*}
E=\begin{pmatrix}
0 & 0 & -1 & 0\\
0 & 0 & 0 & 1\\
-1 & 0 & 0 & 0\\
0 & 1 & 0 & 0
\end{pmatrix}.
\end{equation*}

In order to prove that the start product $\spn{n}{m}$ is $\mathfrak F_m$-invariant, we first give the explicit relation between the $\mathrm{SO}(n+1)$ action of rotations on the quadric $\Qua^n$, $n=2,3,5$ (whose elements are expressed in term of the map $\rho_{(n,m)}$) and the action of the group $\mathfrak F_m$ on $\mathbb C^m$, $m=2,4,8$, respectively.

\begin{proposition} \label{invariance rho}
	Let $g\in \mathfrak F_m$ and $\OL(g)$ the action defined above, $m=2,4,8$,
	% be a element in $\mathfrak F_m=\mathrm{SU}(2), \mathrm{SU}(2)\times \mathrm{SU}(2), \mathrm{SU}(4)$ for $m=2,4,8$ respectively. Let $\mathrm L_g$ be the action of the group $\mathfrak F_m$ defined by Eqs. (\ref{action SU(2) C2}), (\ref{action SU(2)xSU(2) C4}) and (\ref{action SU(4) C8}) for the cases $m=2,4,8$ respectively.
	then exist $R \in \mathrm{SO}(n+1)$, $n=2,3,5$ respectively, such that
	\begin{equation}\label{invariance rho equation}
	R\rho_{(n,m)}(\vecz)=\rho_{(n,m)}(\OL(g) \vecz), \quad \forall \vect z \in \mathbb C^m,
	\end{equation}
	where for $\vect{w}\in \mathbb C^{n+1}$, $R\vect{w}=R \Re(\vect{w})+\imath R\Im(\vect{w})$ with $R\Re(\vect w)$ and $R \Im(\vect w)$ denoting the usual action of $R$ on the real and imaginary part of $\vect{w}$, respectively (regarded as elements of $\mathbb R^{n+1}$).
\end{proposition}
\begin{proof}
	The main idea is the same for the three cases $n=2,3,5$. We describe in detail the most complicated case $n=5$. The cases $n=2,3$ follow in a similar way and we will only sketch the structure of the proof.
	
	The case $n=5$: let us write the map $\rho_{(5,8)}(\vecz)= \big(\rho_1(\vecz),\rho_{2}(\vecz),\ldots,\rho_6(\vecz)\big)$	
	in matrix form
	\begin{equation}\label{rho 5,8 expresed matrix A}
	\rho_\ell(\vecz)=(z_1,z_2,z_3,z_4) \mathcal A_\ell \begin{pmatrix}
	z_5\\ z_6\\z_7\\z_8
	\end{pmatrix}\;,\hspace{0.5cm}
	\begin{array}{l}
	\vect{z}=(z_1,\ldots,z_8),\\[0.1cm]
	\ell=1,\ldots,6,
	\end{array}
	\end{equation}
	where the matrices $\mathcal A_\ell$ are defined as follows:

	\begin{equation*}
	\begin{array}{lll}
	\mathcal A_1= \begin{pmatrix}
	0 & -\imath & 0 & 0\\
	\imath & 0 & 0 & 0\\
	0 & 0 & 0 & \imath\\
	0 & 0 & -\imath & 0\\
	\end{pmatrix},
	&
	\mathcal A_2= \begin{pmatrix}
	0 & 1 & 0 & 0\\
	1 & 0 & 0 & 0\\
	0 & 0 & 0 & 1\\
	0 & 0 & 1 & 0\\
	\end{pmatrix},
%	\\[1cm]
	&
	\mathcal A_3= \begin{pmatrix}
	-1 & 0 & 0 & 0\\
	0 & 1 & 0 & 0\\
	0 & 0 & 1 & 0\\
	0 & 0 & 0 & -1\\
	\end{pmatrix},
	\\[1cm]
	\mathcal A_4= \begin{pmatrix}
	-\imath & 0 & 0 & 0\\
	0 & -\imath & 0 & 0\\
	0 & 0 & \imath & 0\\
	0 & 0 & 0 & \imath\\
	\end{pmatrix},\qquad
%	\\[1cm]
	&
	\mathcal A_5= \begin{pmatrix}
	0 & 0 & 0 & -\imath\\
	0 & 0 & -\imath & 0\\
	0 & -\imath & 0 & 0\\
	-\imath & 0 & 0 & 0\\
	\end{pmatrix}, \qquad
	&
	\mathcal A_6= \begin{pmatrix}
	0 & 0 & 0 & 1\\
	0 & 0 & 1 & 0\\
	0 & -1 & 0 & 0\\
	-1 & 0 & 0 & 0\\
	\end{pmatrix}.
	\end{array}
	\end{equation*}
	\begin{comment}
	Let us consider the following action of $\mathrm{SU}(4)$ on $\mathbb C^8$: given  $U \in \mathrm{SU}(4)$ we define
	\begin{equation}\label{action SU(4) C8}
	\begin{pmatrix}
	z_1'\\[0.1cm]z_2'\\[0.1cm]z_3'\\[0.1cm]z_4'
	\end{pmatrix} = U 
	\begin{pmatrix}
	z_1\\[0.1cm]z_2\\[0.1cm]z_3\\[0.1cm]z_4
	\end{pmatrix}, \hspace{0.5cm}
	\begin{pmatrix}
	z_5'\\[0.1cm]z_6'\\[0.1cm]z_7'\\[0.1cm]z_8'
	\end{pmatrix} = EUE 
	\begin{pmatrix}
	z_5\\[0.1cm]z_6\\[0.1cm]z_7\\[0.1cm]z_8
	\end{pmatrix},%\hspace{0.5cm} U \in \mathrm{SU}(4),
	\end{equation} 
	where the matrix $E$ is defined by
	\begin{equation*}
	E=\begin{pmatrix}
	0 & 0 & -1 & 0\\
	0 & 0 & 0 & 1\\
	-1 & 0 & 0 & 0\\
	0 & 1 & 0 & 0
	\end{pmatrix}.
	\end{equation*}
	\end{comment}
	
	From Eqs. (\ref{action m=2,4,8}), (\ref{action SU(4) C8}) and (\ref{rho 5,8 expresed matrix A})
	\begin{equation*}
	\rho_\ell(\vect{z}')=(z_1,z_2,z_3,z_4) U^t \mathcal A_\ell E U E \begin{pmatrix}
	z_5\\z_6\\z_7\\z_8
	\end{pmatrix}.
	\end{equation*}
	%	where $U^t$ denotes the transpose matrix of $U$.
	
	Let $\mathcal V$ denote the real vector space generated by the matrices $\mathcal A_\ell$, $\ell=1,\ldots,6$. The vector space $\mathcal V$ is the set of complex matrices of the form
	\begin{equation} \label{elements in V}
	\begin{pmatrix}
	-\vartheta & \overline\mu & 0 & \overline \gamma\\[0.1cm]
	\mu & \overline \vartheta & \overline{\gamma} & 0\\[0.1cm]
	0 & -\gamma & \vartheta & \mu\\[0.1cm]
	-\gamma & 0 & \overline \mu & -\overline{\vartheta}
	\end{pmatrix}, \quad \mbox{with $\mu,\vartheta,\gamma \in \mathbb C$.}
	\end{equation}

	We now claim that, for $U \in \mathrm{SU}(4)$, the matrix $U^t \mathcal A_\ell EUE$ is in the vector space $\mathcal V$. To prove this fact, let us denote by $U_{jk}$ the matrix elements of $U$. Since  $\mathrm{det}(U)=1$ and $U^\dagger=U^{-1}$ then by considering the explicit expression for the inverse matrix $U^{-1}$ we find that the matrix elements of $U$ must satisfy the following relations:
	\begin{align*}
	\overline U_{11}\overline U_{23}-\overline U_{21}\overline U_{13} & = U_{42}U_{34}- U_{32}U_{44}\;, &%\\
	\overline U_{11}\overline U_{43}-\overline U_{41}\overline U_{13} & = U_{32}U_{24}- U_{22}U_{34}\;,\\
	\overline U_{11}\overline U_{33}-\overline U_{31}\overline U_{13} & = U_{22}U_{44}- U_{42}U_{24}\;,&%\\
	\overline U_{21}\overline U_{43}-\overline U_{41}\overline U_{23} & = U_{12}U_{34}- U_{32}U_{14}\;,\\
	\overline U_{31}\overline U_{23}-\overline U_{21}\overline U_{33} & = U_{12}U_{44}- U_{42}U_{14}\;,&%\\
	\overline U_{41}\overline U_{33}-\overline U_{31}\overline U_{43} & = U_{12}U_{24}- U_{22}U_{14}\;.
	\end{align*}
	
	Then by using the last equalities and computing the explicit expression for the matrix $U^t \mathcal A_\ell EUE$, we find that $U^t \mathcal A_\ell EUE$ has the form indicated in Eq. (\ref{elements in V}).
	
	The vector spaces $\mathcal V$ is endowed with the real valued inner product 
	\begin{equation*}
	\langle A,B \rangle_\mathcal V = \frac{1}{2}\left(\mathrm{trace}(A B^\dagger)+ \mathrm{trace}(B A^\dagger) \right)\;.
	\end{equation*}
	
	The set of matrices $\left\{\frac{1}{2} \mathcal A_\ell \;|\; \ell=1,\ldots,6 \right\}$ gives an orthonormal basis for the space $\mathcal V$. Thus $U^t \mathcal A_\ell EUE$ must be the following linear combination of the basis elements (summation over repeated indexes):
	\begin{equation*}
	U^t \mathcal A_\ell EUE = \frac{1}{4} \left\langle U^t \mathcal A_\ell EUE, \mathcal A_k\right\rangle_\mathcal V \mathcal A_k\;, \hspace{0.5cm} \ell=1,\ldots, 6.
	\end{equation*}
	
	Therefore we have
	\begin{equation*}
	\rho_\ell(\vect{z}')= (z_1,z_2,z_3,z_4) R_{\ell k} \mathcal A_k\begin{pmatrix}
	z_5\\z_6\\z_7\\z_8,
	\end{pmatrix} = R_{\ell k} \rho_k(\vecz),
	\end{equation*}
	with the real numbers $R_{\ell k}$, $\ell,k=1,\ldots,6$, given by
	\begin{equation}\label{elements R}
	R_{\ell k}= \frac{1}{4} \left\langle U^t \mathcal A_\ell EUE, \mathcal A_k\right\rangle_\mathcal V\;.
	\end{equation}
	
	Thus to each element $U \in \mathrm{SU}(4)$ we associated a $6 \times 6$ matrix $R$ whose matrix elements  are $R_{\ell k}$ given by Eq. (\ref{elements R}). Since the matrix $R$  satisfies the relation $R_{j k} R_{s k}= \delta_{js}$ (thus $R$ must be an orthogonal matrix) then we have a continuous map $U \mapsto R$ from $\mathrm{SU}(4) $ into $\mathrm O(6)$. Since $\mathrm{SU}(4)$ is a connected manifold and the identity element of $\mathrm{SU}(4)$ goes to the identity element of $\mathrm O(6)$ (see Eq. (\ref{elements R})), then the image of $\mathrm{SU}(4)$ under the map we are considering must be the connected component of the identity matrix in $\mathrm O(6)$. Thus $R\in \mathrm{SO}(6)$.
	
	For the cases $n=2,3$, let us write the maps $\rho_{(n,m)}(\vect{z})=(\rho_1(\vect{z}),\ldots,\rho_{n+1}(\vect{z}))$ in the following matrix form
	
	For the case $n=2$:
	\begin{equation*}
	\rho_\ell(\vecz)=(z_1,z_2) \mathcal B_\ell \begin{pmatrix}
	z_1\\z_2
	\end{pmatrix}\hspace{0.5cm}
	\begin{array}{l}
	\vect{z}=(z_1,z_2),\\[0.1cm]
	\ell=1,2,3,
	\end{array}
	\end{equation*}
	where the matrices $\mathcal B_\ell$ are defined as follows:
	\begin{equation*}
	\mathcal B_1=\frac{1}{2}\begin{pmatrix}
	-1 & 0\\
	0 & 1
	\end{pmatrix},\hspace{0.5cm}
	\mathcal B_2=\frac{1}{2}\begin{pmatrix}
	\imath & 0\\
	0 & \imath
	\end{pmatrix},\hspace{0.5cm}
	\mathcal B_3=\frac{1}{2}\begin{pmatrix}
	0 & 1\\
	1 & 0
	\end{pmatrix}.
	\end{equation*}
	
	For the case $n=3$:
	\begin{equation*}
	\rho_\ell(\vecz)=(z_1,z_2) \mathcal C_\ell \begin{pmatrix}
	z_3\\z_4
	\end{pmatrix}\hspace{0.5cm}
	\begin{array}{l}
	\vect{z}=(z_1,z_2,z_3,z_4),\\[0.1cm]
	\ell=1,2,3,4,
	\end{array}
	\end{equation*}
	where the matrices $\mathcal C_\ell$ are defined as follows:
	\begin{equation*}
	\mathcal C_1=\begin{pmatrix}
	1 & 0\\
	0 & 1
	\end{pmatrix},\hspace{0.5cm}
	\mathcal C_2=\begin{pmatrix}
	\imath & 0\\
	0 & -\imath
	\end{pmatrix},\hspace{0.5cm}
	\mathcal C_3=\begin{pmatrix}
	0 & \imath\\
	\imath & 0
	\end{pmatrix},\hspace{0.5cm}
	\mathcal C_4=\begin{pmatrix}
	0 & 1\\
	-1 & 0
	\end{pmatrix}.
	\end{equation*}
	
	In a similar way we did for the case $n=5$ we need to prove, for the case $n=2$, that the matrix $U^t \mathcal B_\ell U$ (for all $U \in \mathrm{SU}(2)$) is in the vector space generated by matrices $\mathcal B_\ell$, $\ell=1,2,3$,  which is not difficult to prove using the parametrization of $\mathrm{SU}(2)$ indicated in Eq. (\ref{g parametrizada}). And for the case $n=3$ we need to prove that the matrix $V^t \mathcal C_\ell W$ (for all $V,W \in \mathrm{SU}(2)$) is in the vector space generated by matrices $\mathcal C_\ell$, $\ell=1,\ldots,4$, which is a consequence of
	$\mathcal C_\ell \in \mathrm{SU}(2)$, $\ell=1,\ldots,4$.
	
	The rest of the proof is similar to the case $n=5$, therefore we will omit it.
\end{proof}

\begin{theorem}
	The star product $\spn{n}{m}$ defined on the algebra $\A{n}{m}$  is $\mathfrak F_m$-invariant in the sense that
	\begin{equation*}
	\big(f_1\circ \OL(g)\big) \spn{n}{m} \big(f_2 \circ \OL(g)\big)=(f_1 \spn{n}{m} f_2) \circ \OL(g),
	%	\OL(g)^*(f_1 \spn{n}{m} f_2)=\big(\OL(g)^*f_1 \big) \spn{n}{m} \big(\OL(g)^*f_2\big), 
	\hspace{0.5cm} g \in \mathfrak F_m,\; \forall f_1,f_2 \in \A{n}{m},
	\end{equation*}
	where $\OL(g)$ is given by the action of the group  $\mathfrak F_m=\mathrm{SU}(2), \mathrm{SU}(2)\times \mathrm{SU}(2), \mathrm{SU}(4)$ on $\mathbb C^2, \mathbb C^4$, $\mathbb C^8$ respectively, indicated in Eq. (\ref{action m=2,4,8}).
\end{theorem}
\begin{proof}

	Given $\tilde R \in \mathrm{SO}(n+1)$, define the operator $\TR_{\tilde R}:L^2(\Sn{n}) \to L^2(\Sn{n})$ by $\TR_{\tilde R} \psi(\vect x)=\psi(\tilde R^{-1}\vect x)$.  Let $A$ be a bounded linear operator with domain in $L^2(\Sn{n})$, $n=2,3,5$ and $g\in \mathfrak F_m$, $m=2,4,8$ respectively. Let $R \in \mathrm{SO}(n+1)$ be the orthogonal matrix mentioned in the hypothesis of Proposition \ref{invariance rho} associated to $g$.
	
	From the expression for the coherent states (see Eq. (\ref{coherent state}))  and Eq. (\ref{invariance rho equation}) we obtain
	\begin{equation}\label{cs=TR cs}
	\cs{\OL(g)\vect z}=\TR_R \cs{\vect z},\quad\forall \vect z \in \mathbb C^m.
	\end{equation}
	%	with $R \in \mathrm{SO}(n+1)$ the orthogonal matrix mentioned in the hypothesis of Proposition \ref{invariance rho}. 
	
	Then from Eqs. (\ref{definition Berezin transform}) and (\ref{cs=TR cs}) we have
	\begin{align}
	\B{n}{m}(A)\circ \OL(g)(\vect z)
	%\OL(g)^*\B{n}{m}(A)(\vect z)
	= \B{n}{m}(\TR_{R^{-1}}A \TR_R)(\vect z),\quad\forall \vect z \in \mathbb C^m. \label{invariance Berezin symbol}
	\end{align}
	Thus $\B{n}{m}(A) \circ \OL(g) $ %$\OL(g)^*\B{n}{m}(A)$ 
	can be expressed as the Berezin symbol of the bounded linear operator $\TR_{R^{-1}}A \TR_R$ and from Eqs. (\ref{extended covariant symbol}) and (\ref{cs=TR cs}) its extended Berezin symbol is
	%Moreover, the extended Berezin symbol of the bounded linear operator $\T_{R^{-1}}A\T_R$ is
	\begin{equation} \label{Berezin symbol extended T_R A T_R}
	\B{n}{m}(\TR_{R^{-1}}A \TR_R)(\vect w,\vect z)=\B{n}{m}(A)(\OL(g)\vect w, \OL(g) \vect z).
	\end{equation}

	Moreover, from Eqs. (\ref{inner product cs = reproducing kernel}), %(\ref{reproducing property Qm}), (\ref{btsn csw = rkf}), 
	(\ref{cs=TR cs}) and the unitary of $\bts{n}$
	\begin{align} \label{invariance reproducing kernel}
	\Q{m}\left(\vect u, \OL(g)\vect z\right) = \left \langle \T_R \csz,\cs{u} \right\rangle_{\Sn{n}}=\Q{m}(\OL(g^{-1})\vect u,\vect z),%,\quad\forall \vect z,\vect u \in \mathbb C^m.
	\end{align}
	where we have used that the orthogonal matrix associated to $ g^{-1} $ which makes Eq. (\ref{invariance rho equation}) holds is $R^{-1}$.

	Thus we conclude from Eqs. (\ref{expression star product}) and (\ref{invariance reproducing kernel}) %for $f_1,f_2 \in \A{n}{m}$ and $\vect z \in \mathbb C^m$
	\begin{align*}
	(f_1 \spn{n}{m} f_2)(\OL(g)\vect z) & =\int\limits_{\mathbb C^m}f_1(\OL(g)\vect z,\vect w) f_2(\vect w, \OL(g)\vect z) \frac{|\Q{m}(\vect w, \OL(g)\vect z)|^2}{\Q{m}(\OL(g)\vect z,\OL(g)\vect z)} \du{m}(\vect w)\\
	& =\int\limits_{\mathbb C^m}f_1\left(\OL(g)\vect z,\OL(g)\vect u\right) f_2\left(\OL(g)\vect u,\OL(g)\vect z\right) \frac{|\Q{m}(\vect u, \vect z)|^2}{\Q{m}(\vect z,\vect z)} \du{m}(\vect u)\\
	& = \big(f_1 \circ \OL(g)\big)\spn{n}{m} \big(f_2 \circ \OL(g)\big)(\vect z)
	\end{align*}
	where we have made the change of variables $\vect u = \OL(g^{-1})\vect w$, used the invariance of the Gaussian measure $\du{m}$ with respect to the action of $\OL(g^{-1})$ and Eqs. (\ref{invariance Berezin symbol}) and (\ref{Berezin symbol extended T_R A T_R}).
\end{proof}

\appendix 
\section{The coherent states are eigenfunctions of an operator}\label{appendix eigenfunctions cs}

In this appendix, we define operators on $L^2(\Sn{n})$ for which the coherent states are their eigenfunctions.
Below we will consider differential operators and we will not go into details about their domains we are considering since they are not of main relevance for our work.

For $\balpha=\rho_{(n,m)}(\vect z)$, $\vect z \in \mathbb C^m$, let us first note that the coherent states $\csalpha$ can be written in terms of the functions $\mathrm{exp(\vect x \cdot \balpha/\hbar)}$. Let 
\begin{equation}\label{appendix 1 eq 7}
\mathbf M= \sqrt {\frac{2}{n-1}} \;\Delta_{\Sn{n}}^{1/4},
\end{equation}
defined by functional calculus, where $\Delta_{\Sn{n}}$ denotes the self-adjoint and normalized spherical Laplacian on the $n$-sphere with discrete spectrum given by $\{(k + \frac{n-1}{2})^2 \;|\; k = 0, 1, 2,\ldots\}$. Since $\Delta_{\Sn{n}}^{-1/4}$ is a continuos operator and 
\begin{equation}\label{appendix 1 eq 2}
\Delta_{\Sn{n}}\left( \vect x \cdot \balpha\right)^\ell = \left( \ell + \frac{n-1}{2}\right)^2 \left (\vect x \cdot \balpha \right)^\ell, \;\;\vect x \in \Sn{n},\; \ell=0,1,\ldots,
\end{equation}
then $\mathbf M^{-1} \csalpha (\vect x)=\mathrm{exp}(\vect x \cdot \balpha/\hbar)$, and therefore 
\begin{equation}\label{appendix 1 eq 6}
\csalpha(\vect x)= \mathbf M \;\mathrm{exp}(\vect x \cdot \balpha/\hbar)
\end{equation}
with $\mathrm{exp}(\vect x \cdot \balpha/\hbar)$ regarded as a function on $\Sn{n}$.

Let $\mathbf E_k$, $k=1,\ldots, n+1$ be the following operators with domain in $L^2(S^n)$
\begin{equation*}
\mathbf E_k=\sum_{j=1}^{n+1}-\imath \tilde{ \vect x}_j\left[-\imath \hbar \tilde{ \mathbf L}_{kj}\right] + \tilde {\vect x}_k \mathbf N\;,
\end{equation*}
where $\tilde{\vect x}_j$ are regarded as multiplicative operators by the coordinates $x_j$ and acting on $L^2(\Sn{n})$, 
\begin{equation*}
\mathbf N= \hbar\left(\sqrt{ \Delta_{\Sn{n}}}-\frac{n-1}{2}\right),\;\;\;\mbox{and}\;\;\;\tilde{\mathbf L}_{kj}= \left(y_k\frac{\partial}{\partial y_j} - y_j \frac{\partial}{\partial y_k}\right)\biggl |_{\Sn{n}}
\end{equation*}
with $y_j$, $j=1,\ldots n+1$ denoting the cartesian coordinates for $\mathbb R^{n+1}$ and for a given operator $A$ on $L^2(\mathbb R^{n+1})$ we denote by $A|_{\Sn{n}}$ the restriction of the operator $A$ to the $n$-sphere. Note that the Laplacian $\Delta_{\Sn{n}}$ is equal to $\tilde{\mathbf L}^2+(\frac{n-1}{2})^2$ with $\tilde{\mathbf L}^2=\sum_{k<j} \tilde{\mathbf L}^2_{kj}$.

Let us consider the angular momentum operators
\begin{equation*}
{ \mathbf L}_{kj}=y_k \frac{\partial}{\partial y_j} - y_j \frac{\partial}{\partial y_k}\;,\;\;\; k,j=1,\ldots,n+1
\end{equation*}
with domain in $L^2(\mathbb R^{n+1})$. 

Introduce spherical coordinates $(r,\boldsymbol \theta)=(r,\theta_1,\ldots,\theta_n)$ for $\mathbb R^{n+1}$ and denote these with the equation $\mathbf y =T(r,\boldsymbol{ \theta})$, where $r$ is the radial coordinate $r=\sqrt{y_1^2+\cdots+y_{n+1}^2}$.

From the chain rule we obtain that for a given smooth function $f(\mathbf y)$,
\begin{align*}
\mathbf L_{kj} f(\vect y) & = \mathbf L_{kj} (f \circ T \circ T^{-1})(\vect y) \\
& = \sum_{q=1}^n \left[y_k \frac{\partial\theta_q}{\partial y_j} - y_j\frac{\partial \theta_q}{\partial y_k}\right] \frac{\partial f \circ T}{\partial \theta_q}(T^{-1}(\vect y)).
\end{align*}
Therefore evaluating both sides of last equation at $\vect y= T(r,\boldsymbol{ \theta})$ and then at $r=1$ we obtain
\begin{equation}\label{appendix 1 eq 1}
\mathbf L_{kj} f(T(1,\boldsymbol{ \theta}))=\tilde{\mathbf L}_{kj} f(T(1,\boldsymbol{ \theta})).
\end{equation}

Note that Eq. (\ref{appendix 1 eq 1}) in particular implies
\begin{equation}\label{appendix 1 eq 4}
\tilde{\mathbf L}_{kj} \mathrm{exp}(\vect x \cdot \balpha/\hbar)= \frac{1}{\hbar} \left(x_k \overline \alpha_j-x_j \overline \alpha_k\right ) \mathrm{exp}(\vect x \cdot \balpha/\hbar)
\end{equation}

On the other hand, from (\ref{appendix 1 eq 2}) and the continuity of the operator $\Delta_{\Sn{n}}^{-1/2}$ we have
\begin{equation*}
\Delta_{\Sn{n}}^{-1/2}\sum_{k=0}^\infty \frac{1}{k !}\left(k+\frac{n-1}{2}\right) \left(\frac{\vect x \cdot \balpha}{\hbar}\right)^k= \mathrm{exp} \left(\frac{\vect x \cdot \balpha}{\hbar}\right),
\end{equation*}
which in turn implies
\begin{equation} \label{appendix 1 eq 3}
\mathbf N \mathrm{exp} \left(\frac{\vect x \cdot \balpha}{\hbar}\right) =(\vect x \cdot \balpha) \mathrm{exp} \left(\frac{\vect x \cdot \balpha}{\hbar}\right).
\end{equation}
Hence, from Eqs. (\ref{appendix 1 eq 4}) and (\ref{appendix 1 eq 3})
\begin{equation}\label{appendix 1 eq 5}
\mathbf E_k \mathrm{exp}\left(\frac{\vect x \cdot \balpha}{\hbar}\right)=\overline \alpha_k \mathrm{exp}\left(\frac{\vect x \cdot \balpha}{\hbar}\right).
\end{equation}

Let $\mathbf A_k$, $k=1,\ldots, n+1$, be the operator \begin{equation}\label{operator A_k}
\mathbf A_k=\mathbf M \mathbf E_k \mathbf M^{-1}.
\end{equation} 
From Eqs. (\ref{appendix 1 eq 6}) and (\ref{appendix 1 eq 5}) we finally conclude that the coherent states $\csalpha$ are eigenfunctions of the operators $\mathbf A_k$, $k=1,\ldots,n+1,$ with eigenvalue $\overline \alpha_k$.
\begin{remark}
It is well known that $L^2(\Sn{n})= \bigoplus \mathcal V_\ell$, where $\mathcal V_\ell$, $\ell=0,1,\ldots,$ is 
the space of spherical harmonics of degree $\ell$ defined as the vector space of restrictions to the $n$-sphere of harmonic homogeneous polynomials of degree $\ell$ defined initially on the ambient space $\mathbb R^{n+1}$ (here by harmonic we mean a function in the kernel of the usual Laplacian operator on $\mathbb R^{n+1}$). 

We assert that the operators $\mathbf A_k$ maps $\mathcal V_\ell$ on $\mathcal V_{\ell-1}$, $\ell>0$, and $\mathbf A_k \varphi=0$ for $\varphi \in \mathcal V_0$. Since the functions in
%Which can be proven if we use that the functions in
 $\{(\balpha \cdot \vect x)^\ell\;|\; \balpha \in \Qua^n\}\subset L^2(\Sn{n})$ are an overcomplete set on $\mathcal V_\ell$ (see \cite{H-81}) we just need to prove the assertion in the functions $(\balpha \cdot \vect x )^\ell$ for $\balpha \in \Qua^n$. From Eqs. (\ref{appendix 1 eq 7}), (\ref{appendix 1 eq 2}) and (\ref{appendix 1 eq 1})
 \begin{align*}
 \mathbf A_k (\balpha\cdot \vect x)^\ell & = \mathbf M \mathbf E_k \left(\frac{2}{n-1}\ell + 1\right)^{-\frac{1}{2}} (\balpha\cdot \vect x)^\ell\\
 & = \hbar \ell \left(\frac{2\ell+n-1}{n-1} \right)^{-\frac{1}{2}} \mathbf M \alpha_k (\balpha\cdot \vect x)^{\ell-1}\\
 & =\hbar \ell \left(\frac{2(\ell-1)+n-1}{2 \ell+n-1}\right)^{\frac{1}{2}} \alpha_k (\balpha\cdot \vect x)^{\ell-1}.
 \end{align*}
  
 Thus the operators $\mathbf A_k, k= 1, \ldots, n + 1$, can be regarded as annihilation operators.
\end{remark}

\section{Necessary results: Stationary phase method and partitioned matrix}\label{Appendix A}

In this appendix, we mention two known results  needed to prove Theorem \ref{theorem star product}. We start with the stationary phase method which, for our purpose, we apply in the following way  (see Theorem 7.7.5 of Ref. \cite{H-90} for details):

\begin{theorem}[Stationary Phase Method]\label{theorem stationary phase method}
	Let $\p$ and $\beta$ be two smooth complex valued functions defined on $\mathbb R^d$ with $d$ any positive integer. Assume that $\beta$ has compact support, $\Im(\p)\ge 0$, $\p$ has a critical point at $\vecx_0$ and $\p'(\vecx)\ne 0$ for $\vecx\ne\vecx_0$ (where $\p'$ denotes the gradient of $\p$). Moreover, assume that $\Im(\p(\vecx_0)) = 0$ and $\mathrm{det} (\p''(\vecx_0))\ne0$ (where $\p''(\vecx_0)$ denotes the Hessian matrix of $\p$ evaluated at the critical point $\vecx_0$). Then
	\begin{equation}\label{stationary phase method}
	\int e^{\imath \p(\vecx)/\hbar} \beta(\vecx) \mathrm d\vecx  = e^{\imath \p(\vecx_0)/\hbar} \left[\mathrm{det}\left(\frac{\p''(\vecx_0)}{2\pi\imath\hbar}\right)\right]^{-\frac{1}{2}}\left[ \sum_{\ell<k} \hbar^\ell \mathbf M_\ell \beta(\vecx_0) + \mathrm O(\hbar^{k})\right]
	\end{equation}
	%\begin{align}\label{stationary phase method}
	%\int_{\mathbb R^d}\mathrm{exp}\left(\frac{\imath}{\hbar}f(\vecx)\right) \beta(\vecx) \mathrm d\vecx & = \mathrm{exp} \left(\frac{\imath}{\hbar}f(\vecx_0)\right) \left[\mathrm{det}\left(\frac{f''(\vecx_0)}{2\pi\imath\hbar}\right)\right]^{-\frac{1}{2}}\nonumber \\[0.2cm]
	%& \hspace{1cm}\times \left [\sum_{\ell<k} \hbar^\ell \mathrm M_\ell \beta(\vecx_0) + \mathrm O(\hbar^{k+1})\right]
	%\end{align}
	where $\mathrm d\vecx$ denotes the Lebesgue measure on $\mathbb R^d$ and
	\begin{equation*} 
	\mathbf M_\ell \beta(\vecx_0)= \left.\sum_{s=\ell}^{3\ell} \frac{\imath^{-\ell} 2^{-s}}{ s!(s-\ell)!} %\left\langle -(\p''(\vecx_0))^{-1}\hat D,\hat D\right\rangle
	\left[ \left(-(\p''(\vecx_0))^{-1}\right) \hat D \cdot \hat D\right]^s \beta (\p_{\vecx_0})^{s-\ell}\right|_{\vecx=\vecx_0}\;,
	\end{equation*}
	with $(\p''(\vecx_0))^{-1}$ the inverse of the matrix $\p''(\vecx_0)$, $\hat D$ the column vector of size $d$ whose $j$ entry is $\partial/\partial_{x_j}$, and
	\begin{equation*}
	\p_{\vecx_0}(\vecx)=\p(\vecx)-\p(\vecx_0)-\frac{1}{2}\p''(\vecx_0)(\vecx-\vecx_0)\cdot (\vecx-\vecx_0).
	\end{equation*}
\end{theorem}

%The second result we mentioned is the Schurt complement of a matrix block. We use this result in Theorem \ref{theorem star product} to avoid laborious calculations by obtaining both the determinant and the inverse matrix.

The second result we mentioned is used in Theorem \ref{theorem star product} to avoid laborious calculations by obtaining both the determinant and the inverse matrix of a partitioned matrix.

\begin{lemma}
	Let $\mathcal A, \mathcal B, \mathcal C, \mathcal D$ be matrices of $\ell \times \ell$, $\ell \times s$, $s \times \ell$, $s \times s$ respectively, and $\mathcal A$ invertible. Then 
	\begin{align}
	\det \begin{pmatrix}
	\mathcal A & \mathcal B \\
	\mathcal C & \mathcal D \\
	\end{pmatrix} %& =\det\left[
	& =\det(\mathcal A) \det(\mathcal D-\mathcal {CA}^{-1}\mathcal B)\label{determinant}.
	\end{align}
	Furthermore, if we assume that  $\mathcal D$ is nonsingular, $\mathcal B^t \mathcal B=0$ and $\mathcal A=\mathbf I_\ell$ where $\mathbf I_\ell$ denotes the  identity matrix of size $\ell$. Then,
	\begin{align}
	\begin{pmatrix}
	\mathcal A & \mathcal B \\
	\mathcal B^t & \mathcal D \\
	\end{pmatrix}^{-1} = \begin{pmatrix}
	\mathbf I_\ell + \mathcal{BD}^{-1}\mathcal B^t & -\mathcal{BD}^{-1} \\
	-\mathcal D^{-1}\mathcal B^t & \mathcal D^{-1} \\
	\end{pmatrix}^{-1}\label{inverse}.
	\end{align}
\end{lemma}
\begin{proof}
	The first part of this Lemma (Eq. (\ref{determinant})) is a direct consequence from proposition 2.8.3 of Ref. \cite{B-09}.
	To prove the second part of this Lemma (Eq. \ref{inverse}), we used  the following analytical inversion formula for a partitioned matrix, provided that $\mathcal A-\mathcal {BD}^{-1}\mathcal B^t$ is nonsingular (see proposition 2.8.7 of Ref. \cite{B-09})
	%	To prove the second part of this Lemma, assume that $\mathcal A-\mathcal {BC}^{-1}\mathcal B^t$ is nonsingular. Using  the following analytical inversion formula for a partitioned matrix (see Fact 2.17.3 of Ref. \cite{B-09})
	%	Let us use that a block matrix can also be inverted using the following analytical inversion formula:
	\begin{align*}
	& \begin{pmatrix}
	\mathcal A & \mathcal B \\
	\mathcal B^t & \mathcal D \\
	\end{pmatrix}^{-1} =%\\
	%& \hspace{0.8cm}
	\begin{pmatrix}
	\left(\mathcal A-\mathcal B\mathcal D^{-1}\mathcal B^t\right)^{-1} & -\left(\mathcal A-\mathcal B \mathcal D^{-1}\mathcal B^t \right)^{-1}\mathcal B \mathcal D^{-1}\\
	-\mathcal D^{-1}\mathcal B^t \left(\mathcal A-\mathcal {BD}^{-1}\mathcal B^t\right)^{-1}
	&\mathcal{D}^{-1}+\mathcal D^{-1}\mathcal B^t \left(\mathcal A-\mathcal {BD}^{-1}\mathcal B^t\right)^{-1} \mathcal B \mathcal D^{-1}
	\end{pmatrix},
	\end{align*}
	and the fact that $(\mathbf I_\ell - \mathcal{BD}^{-1}\mathcal B^t)^{-1}= \mathbf I_\ell + \mathcal{BD}^{-1}\mathcal B^t$.%, we obtain  Eq. (\ref{inverse}).
\end{proof}

%\section*{References}

\newpage

Let $\balpha=\rho_{(n,m)}(\vect z)$ with $\vect z \in \mathbb C^m$. We first proof that the coherent states are eigenfunctions of an operator $\mathbf D=(D_1,\ldots, D_{n+1})$ with eigenvalue $\overline \balpha=(\overline \alpha_1,\ldots,\overline \alpha_{n+1})$, i.e. $B_\ell \csalpha=\overline \alpha_\ell \csalpha$, $\ell=1,\ldots,n+1$.

\end{document}